\documentclass{article}
\usepackage{amsmath, amssymb, amsthm, bm, graphicx}
\usepackage{url}
\usepackage{subcaption}
\usepackage{booktabs}
\usepackage[margin=1in]{geometry}
\usepackage{natbib}
\usepackage{enumitem}
\usepackage{hyperref}
\usepackage{multirow}
\usepackage{xcolor}
\usepackage{algorithm}
\usepackage{algpseudocode}
\newtheorem{proposition}{Proposition}
\newtheorem{corollary}[proposition]{Corollary}
\newcommand{\mv}[1]{{\boldsymbol{\mathrm{#1}}}}
\newcommand{\E}{\mathbb{E}}

\title{Bivariate geostatistical latent variable models for the analysis of antibody density data}

\author{
Emanuele Giorgi$^{1,2}$\thanks{Corresponding author.
Mailing address: Department of Applied Health Science, University of Birmingham,
Edgbaston, Birmingham B15 2TT, UK.
E-mail: \href{mailto:e.giorgi@bham.ac.uk}{e.giorgi@bham.ac.uk}} \and
Jonas Wallin$^{3}$
}

\date{
$^{1}$Department of Applied Health Science, University of Birmingham, Birmingham, UK\\
$^{2}$Faculty of Health and Medicine, Lancaster University, Lancaster, UK\\
$^{3}$Department of Statistics, Lund University, Lund, Sweden\\
}

\begin{document}

\maketitle

\begin{abstract}
The increasing availability of serosurveys that measure antibody
responses to multiple antigens requires the development of methods
that can exploit the full information content of such data, both
biological and spatial. However, the non-Gaussian and potentially
multimodal distributional behaviour of antibody responses makes the
development of such methods inherently complex, especially in a
multivariate setting. Here, we extend the latent variable framework of \citet{giorgiwallin2026}, in which continuous antibody concentrations
are modelled through an individual-level latent seroreactivity process that represents
the level of immune activation to a given antigen. We focus primarily on the bivariate setting and set
out a series of guiding principles that justify the resulting joint
modelling structure. The proposed model captures distinct sources of correlation between
antibody responses, arising both from shared exposure to the same
environment and from biological processes occurring within the same
host. Spatial dependence is introduced through a novel
bivariate Mat\'ern random field, which we use to construct a
parsimonious class of cross-covariance functions between
antigen-specific spatial processes. We illustrate the application of
the framework to analyse data on bivariate antibody measurements from a malaria serosurvey in the Kenyan highlands. Results from the application and a simulation study show that ignoring this correlation substantially degrades inference on joint properties of the antibody distributions and on individual-level seroreactivity, but matters less when interest lies exclusively in each antibody's marginal distribution. Finally, we discuss how the framework could be extended to settings with more than two antigens, and highlight the modelling challenges that arise as the number of antigens grows.
\end{abstract}


\section{Introduction}
\label{sec:intro}
Population-based serological surveys enable quantification of
past exposure and immunity at the individual and population level,
complementing information that clinical case reporting alone
cannot provide \citep{metcalf2016, corran2007}.
The usefulness of such surveys spans a wide range of epidemiological contexts, from
detecting subtle signs of resurgence in low-transmission or
post-elimination settings, to quantifying the accumulation of
immunity across age groups in high-transmission areas, to
characterising infection burden in diseases with predominantly
asymptomatic transmission where case counts systematically
underestimate true exposure \citep{Drakeley2005, arnold2018}.
Advances in serological platforms, from single-antigen
enzyme-linked immunosorbent assays to high-throughput multiplex
bead-based assays, have made it increasingly feasible to measure
antibody responses against multiple antigens simultaneously in
large cross-sectional surveys \citep{Hay2024Serodynamics}.
This has enabled the development of integrated
serosurveillance strategies that can simultaneously monitor exposure
to multiple antigens within a single survey, providing a more
efficient and comprehensive understanding of population immunity than antigen-by-antigen analyses \citep{arnold2018, Carcelen2025}.
Realising this potential, however, requires statistical methods that can jointly analyse multivariate antibody outcomes while accounting for the dependence structure of such data. This dependence operates at two distinct levels. Within an individual, responses to different antigens are correlated because a single exposure event stimulates several simultaneous responses and because shared immunological mechanisms govern their subsequent dynamics. Between individuals, correlation arises from spatial proximity, since individuals living close to one another share the environmental and ecological determinants of exposure. 

The dominant paradigm in serological data analysis reduces quantitative antibody measurements to a binary seropositive or seronegative classification \citep{corran2007,arnold2018}, reflecting a broader
tendency in clinical research to dichotomize continuous data
\citep{royston2006}.  This approach has thus shaped how different sources of dependence have been addressed. Within the model-based geostatistics framework \citep{DiggleGiorgi2016},
generalized linear mixed models embedding a Gaussian spatial random field
have been used to map seroprevalence and delineate foci of transmission
for a range of pathogens, including malaria \citep{Stresman2017},
trachoma \citep{Sasanami2023}, and lymphatic filariasis
\citep{CadavidRestrepo2023}. However, no geostatistical framework has yet been developed to model
the cross-correlation in space across multiple antibody measurements. In a non-spatial context, some studies have proposed methods to handle multivariate antibody outcomes.
\citet{Lubyayi2021} developed a pairwise joint linear mixed model for multiple continuous immunological outcomes measured longitudinally, and \citet{ODriscoll2025} proposed a semi-mechanistic multivariate mixture model fitted to continuous log-transformed antibody titres with the specific aim of jointly inferring pathogen-specific infection prevalence and between-pathogen antibody cross-reactivity in settings where
antigenically related pathogens co-circulate. However, both rely on Gaussian, or log-Gaussian, distributional assumptions which may not adequately capture the skewness and
multimodality that often characterise antibody concentration data, and neither incorporates spatial dependence. To the best of our knowledge, no existing method addresses cross-antigen correlation, age-dependent variation, and spatial dependence among individuals within a single modelling framework that also allows for non-Gaussian, multimodal distributions of the outcome data. This is the gap we address in this paper.

To achieve this, we build on the latent variable framework of
\citet{giorgiwallin2026}, which introduces a latent seroreactivity
variable that captures the continuum of immunological response
across individuals and accounts for the non-Gaussian behaviour
typical of these data. Here, we extend this framework in two
directions. First, we develop a bivariate latent variable
model in which cross-antigen dependence is represented through the
joint distribution of antigen-specific seroreactivity states.
This dependence reflects biological processes operating at
the individual level, whereby a single exposure event
simultaneously stimulates immune responses to multiple
antigens and shared immunological mechanisms within the
host inducing correlation in seroreactivity that is
intrinsic to the immune response itself.
Second, we introduce spatial correlation through a bivariate
Mat\'{e}rn random field, to account for residual correlation
among individuals who share similar exposure environments as a result of spatial proximity. We argue that explicitly accounting for both sources of correlation, host-induced and spatial, through the proposed joint modelling approach is a crucial first step towards a statistical framework that makes full use of serological data, and provides a foundation for more efficient integrated serosurveillance strategies that exploit the rich dependence structure present in multi-antibody data. 

The remainder of the paper is organised as follows.
Section~\ref{sec:framework} recalls the univariate latent variable
model of \citet{giorgiwallin2026} and sets out the principles
guiding its multivariate extensions. Section~\ref{sec:model}
introduces the bivariate latent variable model in a non-spatial
context. Section~\ref{sec:spatial} develops the geostatistical extension,
introducing a bivariate Mat\'{e}rn random field, its SPDE representation for computationally feasible
inference, and epidemiologically relevant population-level predictive targets.
Section~\ref{sec:application} presents an application on the joint analysis of two antibody measurements from a cross-sectional malaria
survey conducted in the Kenyan highlands. Section~\ref{sec:simulation}
presents a simulation study assessing the consequences of ignoring
either spatial or cross-antigen dependence.
Section~\ref{sec:discussion} concludes with a discussion of
limitations and extensions.

Code implementing the analysis pipeline, illustrated on a simulated dataset with the same structure as the data of Section~\ref{sec:application}, is available at
\url{github.com/giorgistat/mlatv-paper}.

\section{Guiding principles for developing a multivariate latent seroreactivity modelling framework}
\label{sec:framework}

\citet{giorgiwallin2026} introduced a latent variable framework
in which each individual's immune activation state is represented
by a latent variable $T \in [0,1]$, referred to as seroreactivity,
with values near zero corresponding to low-activation states
characterised by minimal or absent serological activity and values
near one indicating strong or boosted antibody responses. The
observed antibody concentration is modelled conditionally on $T$ as
\begin{equation}
\label{eq:univ}
  Y \mid T = t \;\sim\;
  \mathcal{N}\!\left(
    (1-t)\mu_0 + t\mu_1,\;
    (1-t)\sigma_0^2 + t\sigma_1^2
  \right),
\end{equation}
so that both the mean and variance interpolate linearly between
boundary values $(\mu_0, \sigma_0^2)$ and $(\mu_1, \sigma_1^2)$
corresponding to the lower and upper extremes of the assay range.
The distribution of $T$ on $(0,1)$ can be specified
as a single parametric distribution, with the Beta family as a
natural choice given its ability to represent shapes ranging from
U-shaped to unimodal. Alternatively, the distribution of $T$ can be
specified as a finite mixture of component densities on $(0,1)$,
each representing a distinct immunological subpopulation with
characteristic seroreactivity levels. In this formulation,
the mixing probabilities vary with age and can be linked
either to mechanistic models of antibody acquisition, such
as catalytic models \citep{Drakeley2005} that describe the rate at
which individuals develop measurable immune responses, or to
flexible regression-based specifications that capture empirical
age-antibody patterns without imposing strong biological
assumptions.

Extending \eqref{eq:univ} to several antigens requires deciding
where in the model hierarchy, each source of dependence should be
introduced. Without constraints on this choice, the resulting
specifications risk being weakly identifiable and difficult to
reconcile with the underlying biology. We therefore set out three
principles that govern the model formulation, stated here in general
terms and applied in Section~\ref{sec:model} to the non-spatial
bivariate model and in Section~\ref{sec:spatial} to its
geostatistical extension.

\begin{itemize}
    \item \textbf{P1.} \emph{Dependence is carried by the latent
    process, not by the observation model.} The parameters of the
    observation model -- i.e. the distribution of $Y \mid T$ -- are taken to reflect fixed characteristics of
    the assay, so that, conditional on the latent seroreactivity
    levels, residual variation in the observed antibody
    concentrations arises from measurement noise and laboratory
    artefacts that do not induce systematic correlation across
    antigens.

    \item \textbf{P2.} \emph{Exposure and response magnitude are modelled separately.} The force of infection to which an
    individual is exposed affects the likelihood of mounting a
    high-seroreactivity response, whilst individual-specific
    immunological processes govern the magnitude of the response
    conditional on exposure.

    \item \textbf{P3.} \emph{Sources of correlation that operate at
different scales should be represented separately.} Dependence
arises from at least three mechanisms acting at distinct levels: a
single infection event stimulating several antigen-specific
responses within an individual, a shared exposure environment
correlating infection risk across individuals living in proximity,
and shared within-host processes correlating response magnitude
independently of any single event. Whenever possible, each mechanism should be captured by its own
correlation structure, so that the underlying processes remain
mutually distinguishable.
\end{itemize}

Taken together, these principles constrain the construction of the
multivariate framework. Following \textbf{P2}, a mixture
formulation for the seroreactivity process is the more natural
choice, since exposure, including that driven by spatial risk
factors, can be accounted for through the mixing probabilities
governing the level of the seroreactivity response, whilst
response magnitude is expressed through the means of the mixture
components. \textbf{P3} in turn points to a hierarchical model
structure in which each mechanism of correlation is represented by
a separate component and therefore retains a clear interpretation.
Of these mechanisms, the shared infection event is the most
pertinent when the antigens derive from the same pathogen and
life-cycle stage, as is the case for the two antigens analysed in
Section~\ref{sec:application}. 

In the remainder of the paper we develop this construction for the
bivariate case, which is sufficient to expose the modelling and
computational challenges that also arise in a general multivariate
treatment. We return in Section~\ref{sec:discussion} to how the
framework extends beyond two antigens.

\section{A bivariate latent variable model in a non-spatial setting}
\label{sec:model}

We now apply the principles of Section~\ref{sec:framework} to develop a modelling framework for
the case of two antigens, in a non-spatial setting. Following \textbf{P1}, we assume that the observed antibody
concentrations are conditionally independent across antigens given
the latent seroreactivity process, so that cross-antigen
dependence is introduced through its joint distribution.

Let $\mv{Y}_i := (Y_{i,1}, Y_{i,2})^\top$ denote the random vector of log-antibody concentrations for individual $i$, and let $\mv{y}_i := (y_{i,1}, y_{i,2})^\top$ be its realised value. For antigen $k \in \{1,2\}$, let $T_{i,k}\in(0,1)$ denote the corresponding latent seroreactivity, and write $\mv{T}_i := (T_{i,1}, T_{i,2})^\top$. Conditional on $\mv{T}_i=\mv{t}$, we define
\begin{align}
\label{eq:obs_mean_var}
\mu_k(t_{i,k}) &:= \mu_{k,0}+t_{i,k}(\mu_{k,1}-\mu_{k,0}), \nonumber\\
\sigma_k^2(t_{i,k}) &:= \sigma_{k,0}^2+t_{i,k}(\sigma_{k,1}^2-\sigma_{k,0}^2),
\end{align}
and assume the conditional joint density
\begin{equation}
\label{eq:obs_joint}
p(\mv{y}_i \mid \mv{T}_i=\mv{t})
=
\prod_{k=1}^{2}
\mathcal{N}\!\left\{
  y_{i,k};\;\mu_k(t_{i,k}),\;\sigma_k^2(t_{i,k})
\right\}.
\end{equation}
Thus, both the conditional mean and the conditional variance interpolate linearly between the lower and upper extremes of the antigen-specific latent scale. As in \citet{giorgiwallin2026}, the parameters $\mu_{k,0}$ and $\mu_{k,1}$ describe the expected antibody concentration at the lower and upper boundaries of detectable seroreactivity for antigen $k$. The constraint $\mu_{k,1}>\mu_{k,0}$ fixes the orientation of the latent scale, so that larger values of $T_{i,k}$ correspond to larger expected antibody responses. The parameters $\sigma_{k,0}^2$ and $\sigma_{k,1}^2$ play the analogous role for the conditional measurement variance.

Under \eqref{eq:obs_joint}, cross-antigen dependence is therefore
represented entirely by the joint density of $\mv{T}_i$, which we
denote by $g(\mv{t}\mid\mv{d}_i)$. The vector $\mv{d}_i$ may
include an intercept, age transformations, and other covariates,
but we write the observed age separately as $a_i$ whenever the age
dependence of a particular parameter is highlighted.

In a bivariate context, the unit square $(0,1)^2$ is the latent space for the joint immune profile. Points near the origin correspond to low seroreactivity for both antigens, points near $(1,1)$ correspond to strong responses to both antigens, and mixed configurations represent individuals with high seroreactivity to one antigen but low seroreactivity to the other. To model the bivariate seroreactivity process, we use a finite mixture of truncated Gaussian components on the unit square. We argue that finite mixtures of truncated bivariate Gaussians provide sufficient flexibility to capture complex patterns in the data, while remaining computationally more tractable than alternative specifications such as bivariate Beta-type distributions.

Let $\mv{Z}_i := (Z_{i,1},Z_{i,2})^\top\in\{0,1\}^2$ denote the component indicators, where $0$ denotes the lower seroreactivity component and $1$ denotes the upper seroreactivity component for the corresponding antigen. Conditional on $\mv{Z}_i=\mv{z}$, with $\mv{z} := (z_1,z_2)^\top$, and on $\mv{d}_i$, we model the latent seroreactivity pair by
\begin{equation}
\label{eq:biv_latentT_component}
\mv{T}_i \mid \mv{Z}_i=\mv{z},\,\mv{d}_i
\;\sim\;
\mathcal{N}_{(0,1)^2}\!\left\{
m(\mv{d}_i;\, \mv{z}_i),\,\Sigma_T
\right\}.
\end{equation}
Here, $\mathcal{N}_{(0,1)^2}\{\mv{m},\Sigma\}$ denotes a Gaussian distribution with location parameter $\mv{m}$ and covariance matrix $\Sigma$, restricted to $(0,1)^2$ and renormalised. The vector $m(\mv{d}_i;\, \mv{z}_i)$ is therefore a pre-truncation location parameter for the component, not necessarily the exact mean of the truncated distribution. We define
\begin{equation*}
m(\mv{d}_i;\, \mv{z}_i)
:=
\left\{m_1(\mv{d}_i;\, z_{i,1}),\,m_2(\mv{d}_i;\, z_{i,2})\right\}^\top,
\end{equation*}
where, for each antigen $k\in\{1,2\}$, the lower and upper latent component locations are
\begin{align}
\label{eq:biv_latentT_location_links}
m_k(\mv{d}_i;\, 0)
&:= -2 + 4\,h(\mv{d}_i^\top\alpha_{k,0}), \nonumber\\
m_k(\mv{d}_i;\, 1)
&:= m_k(\mv{d}_i;\, 0)
 + \{2-m_k(\mv{d}_i;\, 0)\}\,h(\mv{d}_i^\top\alpha_{k,1}),
\end{align}
with $h(u) := \bigl(1+\exp(-u)\bigr)^{-1}$, so that $-2<m_k(\mv{d}_i;\,0)<m_k(\mv{d}_i;\,1)<2$. The bounded range $(-2,2)$ is used to avoid weak identifiability in regions where the location parameter is far from the domain of $\mv{t}$, since in this regime the corresponding distribution is close to its limiting exponential form.

The observed bivariate antibody density is obtained by
integrating out the latent seroreactivity pair,
\begin{equation*}
p(\mv{y}_i\mid\mv{d}_i)
=
\int_{(0,1)^2}
p(\mv{y}_i\mid\mv{T}_i=\mv{t})\,
g(\mv{t}\mid\mv{d}_i)
\,d\mv{t}.
\end{equation*}

This integral is approximated by partitioning $(0,1)^2$
into $M^2$ rectangular cells $B_{r,s}$, $r,s=1,\ldots,M$,
with midpoints $\mv{t}_{r,s}$. The implementation uses a
renormalised midpoint rule rather than evaluating bivariate
Gaussian probabilities for every cell. In particular, the
discrete mass assigned to cell $B_{r,s}$ under component
$\mv{z}$ is
\begin{equation*}
\widetilde w_{i,r,s}(\mv{z})
:=
\frac{|B_{r,s}|\,
\phi_2\!\left\{\mv{t}_{r,s};
m(\mv{d}_i;\mv{z}),\Sigma_T\right\}}
{\displaystyle
\sum_{r'=1}^{M}\sum_{s'=1}^{M}
|B_{r',s'}|\,
\phi_2\!\left\{\mv{t}_{r',s'};
m(\mv{d}_i;\mv{z}),\Sigma_T\right\}},
\end{equation*}
where $\phi_2(\cdot;\mv{m},\Sigma_T)$ denotes the
untruncated bivariate Gaussian density. Renormalisation over
the grid supplies the truncation constant numerically. The
marginal density is then approximated by
\begin{equation}
\label{eq:biv_latentT_quadrature}
p(\mv{y}_i\mid\mv{d}_i)
\approx
\sum_{\mv{z}\in\{0,1\}^2}
\pi_i(\mv{z})
\sum_{r=1}^{M}\sum_{s=1}^{M}
\widetilde w_{i,r,s}(\mv{z})\,
p(\mv{y}_i\mid\mv{T}_i=\mv{t}_{r,s}),
\end{equation}
where $\widetilde w_{i,r,s}(\mv{z})$ approximates the
corresponding truncated-Gaussian cell probability. The full
log-likelihood is
\begin{equation*}
\ell(\vartheta)
:=
\sum_{i=1}^n \log p(\mv{y}_i\mid\mv{d}_i;\vartheta),
\end{equation*}
where $\vartheta$ collects all model parameters.

The covariance matrix of the truncated Gaussian components is written as
\begin{equation}
\label{eq:biv_latentT_covariance}
\Sigma_T := D_T R_T D_T,
\qquad
D_T := \operatorname{diag}\{\zeta_1,\zeta_2\},
\qquad
R_T :=
\begin{pmatrix}
1 & \rho_T\\
\rho_T & 1
\end{pmatrix}.
\end{equation}
The scale parameters $\zeta_1$ and $\zeta_2$ are standard deviations that control the within-component spread of the latent seroreactivity coordinates. The parameter $\rho_T\in(-1,1)$ denotes the correlation of the bivariate Gaussian distribution prior to truncation, quantifying the residual continuous dependence that persists after conditioning on the component labels. Hence, it should not be interpreted as the Pearson correlation of the truncated component itself, since the latter also depends on $m(\mv{d}_i;\mv{z}_i)$, $\zeta_1$, and $\zeta_2$. The resulting latent density is
\begin{equation}
\label{eq:biv_latentT_density}
g(\mv{t}\mid\mv{d}_i)
:=
\sum_{\mv{z}\in\{0,1\}^2}
\pi_i(\mv{z})\,
\mathcal{N}_{(0,1)^2}\!\left\{
\mv{t};\,m(\mv{d}_i;\, \mv{z}_i),\,\Sigma_T
\right\},
\qquad \mv{t}\in(0,1)^2,
\end{equation}
where $\pi_i(\mv{z}) := \Pr(\mv{Z}_i=\mv{z}\mid\mv{d}_i)$ is the joint probability of component $\mv{z}$ for individual $i$.

We define the probabilities $\pi_i(\mv{z})$ through a sequential logistic parametrisation, in which the first antigen has a marginal component probability and the second antigen has a conditional component probability given the first. The ordering is a parametrisation of the joint distribution and should not be interpreted as a causal ordering between the two antigens. Specifically, let
\begin{align}
\label{eq:mixing_probs}
p_{i,1}
&:= \Pr(Z_{i,1}=1\mid\mv{d}_i)
 = h(\mv{d}_i^\top\gamma_1), \nonumber\\
p_{i,2}(z_1)
&:= \Pr(Z_{i,2}=1\mid Z_{i,1}=z_1,\mv{d}_i)
 = h(\mv{d}_i^\top\gamma_2+\delta(a_i)z_1).
\end{align}
The joint mixture probability is then
\begin{equation*}
\pi_i(z_1,z_2)
=
 p_{i,1}^{z_1}(1-p_{i,1})^{1-z_1}
 p_{i,2}(z_1)^{z_2}\{1-p_{i,2}(z_1)\}^{1-z_2}.
\end{equation*}
The age-dependent parameter $\delta(a_i)$ controls the association between the two component labels.  Positive values imply that high seroreactivity for the first antigen increases the probability of high seroreactivity for the second antigen, after accounting for $\mv{d}_i$, and viceversa.

The interpretation of $\delta(a_i)$ depends on the antigens being
modelled. When the two antigens belong to the same pathogen, a
single exposure event may stimulate responses to both, so a
positive association between their component labels is natural and
can be enforced by setting
$\delta(a_i) := \exp(\delta_0+\delta_1\log a_i)>0$. We allow this
association to vary with age because antigens differ in
immunogenicity and in the durability of the antibody responses
they elicit, so the two responses need not reach the
high-seroreactivity regime at the same rate. At young ages, when
cumulative exposure is low, an individual may therefore occupy the
high regime for one antigen and the low regime for the other. As
exposure accumulates this mismatch resolves and the two regimes
become more concordant, at a rate captured by $\delta_1$ over and
above the age dependence already carried by $\gamma_1$ and
$\gamma_2$. When the antigens belong to different pathogens, the
direction of association may be unknown and may reflect
co-infection, differential exposure, cross-reactivity, or
immune-mediated interactions, so an unconstrained specification
such as $\delta(a_i) := \delta_0+\delta_1\log a_i$ is preferable,
allowing the data to determine both the sign and the magnitude of
the association. The model therefore separates two of the three
mechanisms identified in \textbf{P3}, with $\delta(a_i)$ encoding
the shared-infection-event pathway and $\rho_T$ encoding the
shared within-host immunological mechanisms, describing the
continuous co-variation in latent seroreactivity that remains once
the regime labels have been fixed. The third mechanism, which corresponds to the shared
exposure environment, is addressed in Section~\ref{sec:spatial}.

Figure~\ref{fig:biv-latent-example} shows an example for a parameter configuration that reflects two positively associated antigens, as in the case of the malaria application of Section~\ref{sec:application}. For the purpose of illustration, we have set all the parameters that express the effect of age to zero. The mixing probability for the first antigen is $p_1 = 0.41$, so roughly equal numbers of individuals are assigned to the low and high seroreactivity components for antigen 1. The cross-antigen association parameter $\delta = 3.55$ induces a strong positive dependence between the two antigen-specific regimes: among individuals in the low seroreactivity component for antigen 1, only $7\%$ are in the high seroreactivity component for antigen 2, whereas among individuals in the high seroreactivity component for antigen 1, this proportion rises to $72\%$. As a result, most probability mass is concentrated in the lower--lower and upper--upper corners of the unit square, with joint probabilities $\pi_i(0,0) \approx 0.55$ and $\pi_i(1,1) \approx 0.30$, respectively, while the off-diagonal regimes have joint probabilities $\pi_i(1,0) \approx 0.11$ and $\pi_i(0,1) \approx 0.04$. The within-component correlation $\rho_T = 0.82$ further reinforces the tendency for the two latent seroreactivity values to move together inside each component. The right panel shows the bivariate distribution of the induced antibody concentrations. We observe that, although most of the probability mass of the joint bivariate distribution lies along the diagonal corresponding to the identity line, a non-negligible amount of probability mass is also concentrated in the lower-right corner of the panel, and a substantially smaller concentration in the upper-left corner.

We also point out an alternative approach for specifying the mixing probabilities
$\pi_i(z_1,z_2)$ of the component indicators $\mv{Z}_i$ that allows to model marginal probabilities directly.  This is based on the use of a \citet{plackett1965} copula governed by a
covariate-dependent odds ratio $\theta(a_i)$, taking the same functional
form as $\delta(a_i)$ in \eqref{eq:mixing_probs}. Writing
$p_{i,k}:=\Pr(Z_{i,k}=1\mid\mv{d}_i)$, the joint probability $\pi_i(1,1)$
is the solution of
\begin{equation*}
\theta(a_i) = \frac{\pi_i(1,1)\,\pi_i(0,0)}{\pi_i(1,0)\,\pi_i(0,1)},
\qquad
\pi_i(1,1)+\pi_i(1,0)=p_{i,1}, \qquad \pi_i(1,1)+\pi_i(0,1)=p_{i,2},
\end{equation*}
which admits a closed form for $\theta(a_i)\neq1$ \citep{plackett1965},
with the remaining three cell probabilities obtained by subtraction from
the margins $p_{i,1}$ and $p_{i,2}$. This is a reparametrisation  of \eqref{eq:mixing_probs}, since the sequential
specification already implies a constant conditional odds ratio
$\exp\{\delta(a_i)\}$ between $Z_{i,1}$ and $Z_{i,2}$, and both
parametrisations describe the same three-dimensional family of joint
distributions with the given margins. The Plackett parametrisation has the advantage that
$\gamma_1$ and $\gamma_2$ each retain a direct marginal interpretation as
the log-odds of high seroreactivity for the corresponding antigen, whereas
under the sequential specification only $\gamma_1$ retains this
interpretation, since the marginal probability implied for the second
antigen is a nonlinear function of $\gamma_1$, $\gamma_2$, and
$\delta(a_i)$. This distinction becomes material beyond the bivariate
case, where multivariate Plackett-type constructions do not admit a
comparably simple closed form, whereas the sequential formulation extends
directly by introducing successive conditional terms $\delta_{jk}(a_i)$.

\begin{figure}[ht!]
\centering
\begin{subfigure}{0.48\textwidth}
\centering
\includegraphics[width=\textwidth]{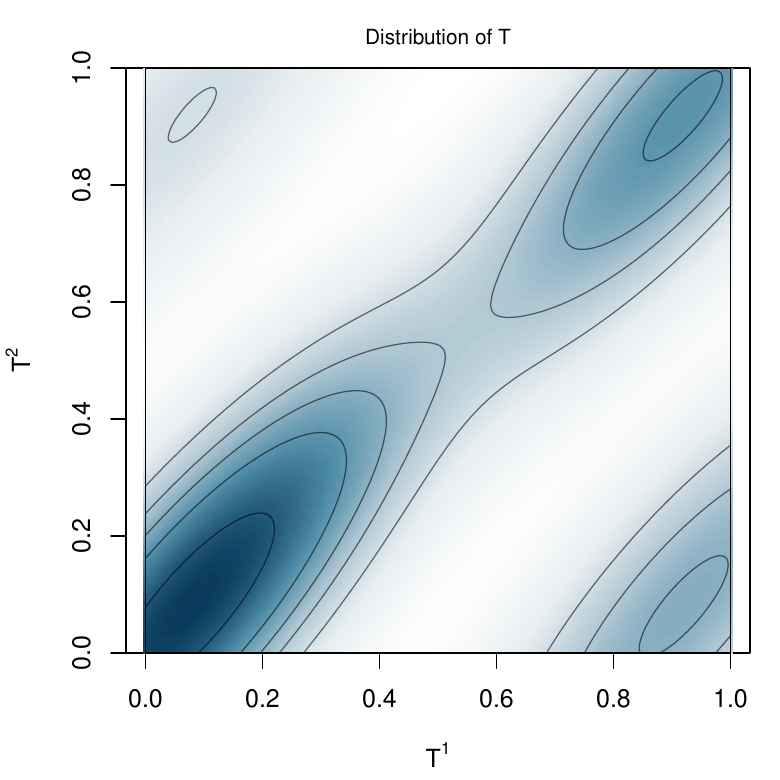}
\caption{}
\end{subfigure}
\hfill
\begin{subfigure}{0.48\textwidth}
\centering
\includegraphics[width=\textwidth]{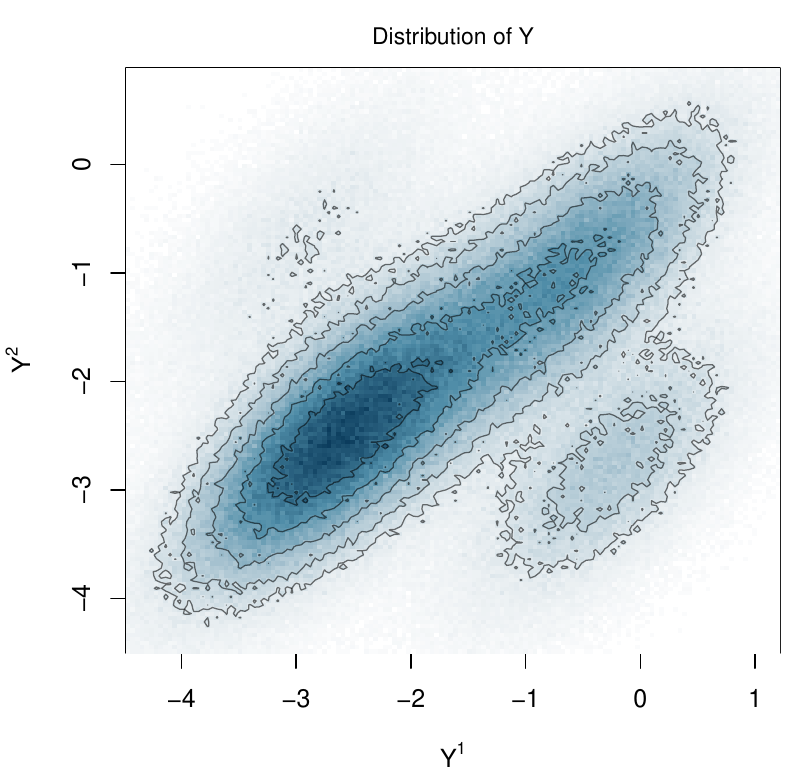}
\caption{}
\end{subfigure}
\caption{Illustration of the bivariate latent seroreactivity model. Panel (a) shows the four-component truncated Gaussian mixture on the latent unit square, for $p_1 = 0.41$, cross-antigen association parameter $\delta = 3.55$, and within-component correlation $\rho_T = 0.82$. Panel (b) shows the bivariate distribution of the observed antibody concentrations induced by this latent sero-reactivity structure.}
\label{fig:biv-latent-example}
\end{figure}

\section{Geostatistical extension}
\label{sec:spatial}

We embed the bivariate latent variable model within a geostatistical
framework by allowing the mixing probabilities $\pi_i(\mv{z})$ to
depend on spatially structured random effects. This follows directly
from \textbf{P1} and \textbf{P2} of Section~\ref{sec:framework}.
By \textbf{P1}, the parameters of the observation
model~\eqref{eq:obs_joint} are held constant across locations, since
they describe assay characteristics rather than exposure. By
\textbf{P2}, spatial variation represents variation in the force of
infection and therefore enters the mixing probabilities governing
regime membership, leaving the magnitude of the response conditional
on regime unchanged. The mixture specification for $\mv{T}_i$
adopted in Section~\ref{sec:model} is what makes this separation
available. Specifically, we extend the sequential logistic
parametrisation in~\eqref{eq:mixing_probs} by introducing a
bivariate Gaussian random field
$\mv{S}(\mv{x}) := \{S_1(\mv{x}), S_2(\mv{x})\}^\top$, so that
\begin{align}
\label{eq:spatial_mixing_probs}
p_{i,1}
&= h(\mv{d}_i^\top \gamma_1 + S_1(\mv{x}_i)), \nonumber \\
p_{i,2}(z_1)
&= h(\mv{d}_i^\top \gamma_2 + \delta(a_i)\, z_1 +
S_2(\mv{x}_i)).
\end{align}
The random field $S_1$ captures residual spatial variation in the
marginal probability of high seroreactivity for the first antigen, while
$S_2$ captures residual spatial variation in the probability of high
seroreactivity for the second antigen conditional on the regime of the
first. The three mechanisms identified in \textbf{P3} are now each
represented by a separate component of the model. The term
$\delta(a_i)\,z_1$ retains its role from Section~\ref{sec:model} of
encoding the shared-infection-event pathway; the fields $S_1$ and
$S_2$, together with their cross-correlation, encode the shared
exposure environment; and the within-component correlation $\rho_T$
in~\eqref{eq:biv_latentT_component} encodes the shared within-host
immunological mechanisms.

\subsection{Bivariate Mat\'{e}rn random fields}
\label{sec:biv-matern}

\begin{figure}[ht!]
    \centering
    \includegraphics[width=\textwidth]{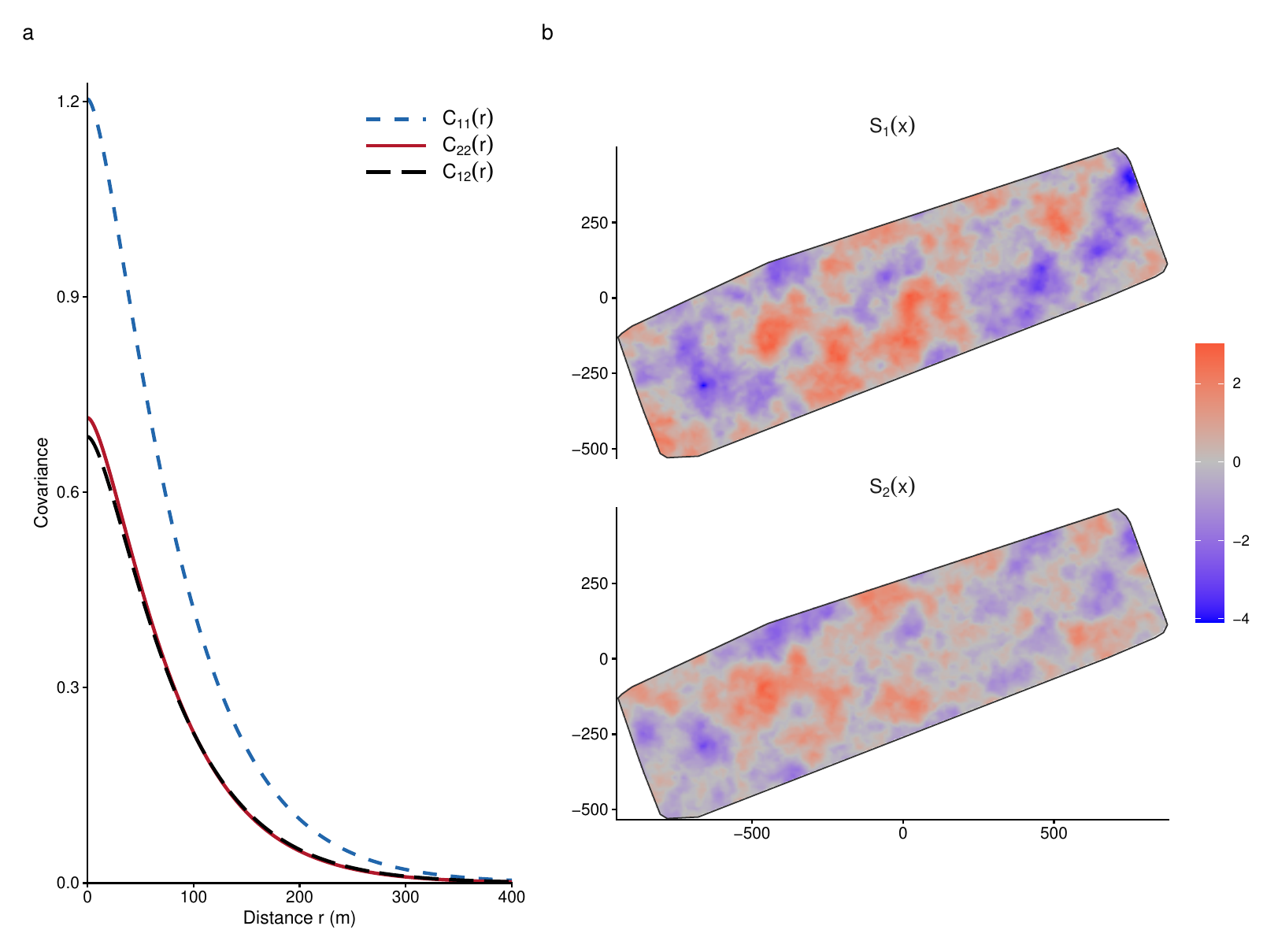}
    \caption{(a) Example of covariance and cross-covariance functions of a bivariate field with $\sigma_1 = 1.1$, $\sigma_2 = 0.9$, $k_1^{-1} = 164.1$ and $k_2^{-1} = 155.39$. The smoothness parameter is fixed at the same value for the two fields at $\nu=1$ (see \eqref{eq:biv-matern-unit-smoothness}). (b) Simulated fields based on the covariance function illustrated in panel (a). The locations of the simulation are defined over a prediction grid within the convex hull of the locations of the application in Section~\ref{sec:application}.}
    \label{fig:crosscov-sim-combined}
\end{figure}

To represent both the marginal spatial variation and the dependence
between the two latent fields, let
$\mv{S}(\mv{x}) := \{S_1(\mv{x}),S_2(\mv{x})\}^\top$ be a mean-zero
stationary Gaussian random field. Its second-order structure is
characterised by the semi positive definite matrix-valued covariance function
\begin{equation*}
\mv{C}(\mathbf{h})
:=
\{C_{ij}(\mathbf{h})\}_{i,j=1}^2,
\qquad
C_{ij}(\mathbf{h})
:=
\mathrm{Cov}\{S_i(\mv{x}),S_j(\mv{x}+\mathbf{h})\}.
\end{equation*}
Stationarity ensures that this covariance depends on the locations
only through the lag $\mathbf{h}$, while isotropy implies that each
entry depends on $\mathbf{h}$ only through the separation
$r := \|\mathbf{h}\|$; see \citet{kleiber2017coherence}
for the matrix-valued covariance formulation, and more details. The diagonal entries
describe the marginal spatial dependence, whereas the off-diagonal
entries describe dependence between the two fields.

We specify the marginal covariance of component $k$ by a Mat\'{e}rn
model with smoothness $\nu_{k}>0$, inverse range $\kappa_k>0$,
and marginal standard deviation $\sigma_k>0$. Although $\kappa_k$
is the natural scale parameter in the covariance and SPDE
representations, spatial dependence is more readily interpreted
through the practical range
$\omega_{k} := \sqrt{8\nu_{k}}/\kappa_k$. We therefore use
$\omega_{k}$ when reporting the spatial scale of component $k$.

The marginal covariance between two locations $\mv{x}$ and
$\mv{x}'$ for field $k$ is
\begin{equation}
\label{eq:matern_cov}
C_{kk}(r)
=
\sigma_k^2\,\mathcal{M}(\kappa_k r;\nu_{k}),
\end{equation}
where
$
\mathcal{M}(z;\nu)
:=
\frac{2^{1-\nu}}{\Gamma(\nu)}z^\nu K_\nu(z)$, with $K_\nu$
denoting the modified Bessel function of the second kind. The
smoothness parameter controls the local regularity of the field,
whereas $\kappa_k$ controls the rate at which spatial dependence
decays with distance.

The triangular multivariate Mat\'{e}rn construction of
\citet{bolin2020multivariate} gives a cross-covariance whose
shape is determined by the two marginal covariance models. Let
$q_k := (\nu_{k}+1)/2$,
$\bar\nu := (\nu_1+\nu_2)/2$, and
$\kappa_u^2 := u\kappa_1^2+(1-u)\kappa_2^2$ for $0<u<1$. The
cross-covariance is
\begin{equation}
\label{eq:biv-matern-crosscov}
C_{12}(r)
=
\rho_S\,\sigma_1\sigma_2\,
\frac{\sqrt{\nu_1\nu_2}\,\Gamma(\bar\nu)\,\kappa_1^{\nu_1}\kappa_2^{\nu_2}}
     {\Gamma(q_1)\Gamma(q_2)}
\int_0^1
u^{q_1-1}(1-u)^{q_2-1}
\kappa_u^{-2\bar\nu}
\mathcal{M}(\kappa_u r;\bar\nu)\,du .
\end{equation}
where $\rho_S\in(-1,1)$ is the cross-field correlation
parameter.

The same covariance family can be represented through the Whittle
SPDE construction. In the univariate case,
\citet{whittle1963stochastic} showed that a stationary
Mat\'{e}rn field can be represented as the solution of
\begin{equation}
\label{eq:univ-spde-whittle}
\tau(\kappa^2-\Delta)^{\alpha/2}S(\mv{x})
=
\mathcal{W}(\mv{x}),
\end{equation}
where $\Delta$ is the Laplacian, $\mathcal{W}$ is Gaussian
white noise, $\alpha := \nu+d/2$, and $\tau$ fixes the marginal
variance. In two dimensions, the variance normalisation is
$\tau^2 := (4\pi\nu\,\sigma^2\kappa^{2\nu})^{-1}$. For the
bivariate model above, $d=2$ gives $\alpha_k=\nu_{k}+1$;
equivalently, component $k$ uses the operator
$\mathcal{L}_k := (\kappa_k^2-\Delta)^{(\nu_{k}+1)/2}$.

The bivariate random field with cross-covariance
\eqref{eq:biv-matern-crosscov} is obtained by coupling the two
univariate operators through a lower-triangular dependence
matrix. We use
\begin{equation*}
\mv{D}(\rho_S)
:=
\begin{pmatrix}
1 & 0\\
 -\frac{\rho_S}{\sqrt{1-\rho_S^2}} & \frac{1}{\sqrt{1-\rho_S^2}} 
\end{pmatrix}.
\end{equation*}
This is the bivariate triangular parametrisation written
directly in terms of the correlation parameter $\rho_S$. Indeed, if
$\mv{R}:=\mv{D}(\rho_S)^{-1}$, then the off-diagonal entry of
$\mv{R}\mv{R}^{\top}$ is $\rho_S$. If $\widetilde\rho_S$ denotes the raw
lower-triangular parameter used in \citet{bolin2020multivariate}, then
$\rho_S := \widetilde\rho_S/\sqrt{1+\widetilde\rho_S^2}$.With
$\tau_k^2 := \bigl(4\pi\nu_{k}\,\sigma_k^2
\kappa_k^{2\nu_{k}}\bigr)^{-1}$, the coupled SPDE system is
\begin{equation}
\label{eq:bispde}
\mv{D}(\rho_S)
\begin{bmatrix}
\tau_1\mathcal{L}_1 & 0\\
0 & \tau_2\mathcal{L}_2
\end{bmatrix}
\begin{bmatrix}
S_1(\mv{x})\\
S_2(\mv{x})
\end{bmatrix}
=
\begin{bmatrix}
\mathcal{W}_1(\mv{x})\\
\mathcal{W}_2(\mv{x})
\end{bmatrix},
\end{equation}
where $\mathcal{W}_1$ and $\mathcal{W}_2$ are independent
Gaussian white-noise processes. The connection between this
SPDE representation and the covariance specification is stated
formally below. 

As the semi-explicit form of the cross-covariance was not previously established, we provide the following proposition:
\begin{proposition}
\label{prop:biv-matern-validity}
Let $\nu_{k}>0$, $\kappa_k>0$, and $\sigma_k>0$ for
$k=1,2$, and let $\rho_S\in(-1,1)$. For a spatial lag
$\mathbf{h} := \mv{x}'-\mv{x}\in\mathbb{R}^2$, put
$r := \|\mathbf{h}\|$ and define
\begin{equation*}
\mv{C}(\mathbf{h})
:=
\begin{pmatrix}
C_{11}(r) & C_{12}(r)\\
C_{12}(r) & C_{22}(r)
\end{pmatrix},
\end{equation*}
where the marginal covariances are given by
\eqref{eq:matern_cov} and the cross-covariance is given by
\eqref{eq:biv-matern-crosscov}. Then $\mv{C}$ is a valid
stationary and isotropic bivariate covariance function.
Moreover, the coupled system \eqref{eq:bispde} admits a
mean-zero stationary Gaussian weak solution, unique in
distribution, whose covariance function is $\mv{C}$.
\end{proposition}

The cross-covariance in \eqref{eq:biv-matern-crosscov} admits two
useful closed-form reductions. When the two fields share a common
inverse-range parameter, it coincides with the parsimonious
bivariate Mat\'{e}rn form of \citet{gneiting2010matern}. A second
closed form is available when the inverse-range parameters differ
but both marginal smoothness parameters are equal to one. These
reductions are stated formally in the following corollary.
The correlation
$\rho:=\operatorname{Corr}\{S_1(\mv{x}),S_2(\mv{x})\}$ need not equal
$\rho_S$ and, for fixed marginal Mat\'{e}rn parameters, may have an attainable
range that is a strict subset of $(-1,1)$; see
Corollary~\ref{cor:biv-matern-correlation} in Appendix~\ref{app:proofs}.
\begin{corollary}
\label{cor:biv-matern-special-cases}
Under the conditions of Proposition~\ref{prop:biv-matern-validity},
and if
$\kappa_1=\kappa_2=\kappa$ , then the cross-covariance has the closed form
\begin{equation}
\label{eq:biv-matern-common-range}
C_{12}(r)
=
\rho_S\,\sigma_1\sigma_2
\frac{\sqrt{\nu_1\nu_2}}{\bar\nu}\,
\mathcal{M}(\kappa r;\bar\nu).
\end{equation}
If $\kappa_1\neq \kappa_2$ and $\nu_1=\nu_2=1$, the cross-covariance has the
closed form
\begin{equation}
\label{eq:biv-matern-unit-smoothness}
C_{12}(r)=
2\rho_S\sigma_1\sigma_2\kappa_1\kappa_2
\frac{K_0(\kappa_1r)-K_0(\kappa_2r)}
{\kappa_2^2-\kappa_1^2}.
\end{equation}
\end{corollary}
The proofs of the Proposition and the Corollary are given in Appendix~\ref{app:proofs}. An example of covariance function and simulated fields based on \eqref{eq:biv-matern-unit-smoothness} are shown in Figures \ref{fig:crosscov-sim-combined}.

\subsubsection{Finite element methods approximation}

The SPDE representation becomes computationally useful after replacing the
continuous fields by finite-dimensional basis expansions. Following the
SPDE--FEM construction of \citet{lindgren2011explicit}, we triangulate the
spatial domain and associate a continuous, piecewise-linear basis function
$\psi_j$ with each of the $m$ mesh vertices. Each basis function equals one at
its own vertex, zero at all other vertices, and has support only on the
triangles meeting that vertex. For component $k$, the field is then
approximated by
\begin{equation*}
S_k(\mv{x})
\;\approx\;
\sum_{j=1}^{m}\psi_j(\mv{x})w_{k,j}
\;=\;
\mv{\psi}(\mv{x})^\top\mv{w}_k,
\end{equation*}
where $\mv{\psi}(\mv{x}) :=
\{\psi_1(\mv{x}),\ldots,\psi_m(\mv{x})\}^\top$ and
$\mv{w}_k := (w_{k,1},\ldots,w_{k,m})^\top$ is the vector of
random basis weights. We construct the constrained Delaunay triangulation and
the corresponding projection matrices using the \texttt{fmesher} package
\citep{lindgren2024fmesher}.

Because the spatial fields in \eqref{eq:spatial_mixing_probs} are latent,
the data may contain limited information about their differentiability.
Moreover, allowing general smoothness values requires rational approximations
of the fractional SPDE operator, as implemented in the \texttt{rSPDE} package
\citep{bolin2024rspde}. We therefore fix the smoothness parameters at
$\nu_1=\nu_2=1$. This choice gives $\alpha_k=2$ in
\eqref{eq:univ-spde-whittle}, the lowest integer operator order available in
two dimensions, and hence the most computationally efficient FEM
representation considered here. It also places the model in the
unit-smoothness case with the explicit cross-covariance stated in
Corollary~\ref{cor:biv-matern-special-cases}.

The distribution of the weights $\mv{w}_k$ is obtained by taking the weak
solution of the component SPDE in \eqref{eq:univ-spde-whittle} and applying
the Galerkin method \citep{lindgren2011explicit}. For $\nu_{k}=1$, this
gives $\mv{L}_k := \tau_k(\kappa_k^2\mv{C}+\mv{G})$ and the sparse precision
matrix
\begin{equation*}
\mv{w}_k
\sim
\mathcal N(\mv{0},\mv{Q}_k^{-1}),
\qquad
\mv{Q}_k
:=
\mv{L}_k^\top\mv{C}^{-1}\mv{L}_k
=
\tau_k^2\!\left(
\kappa_k^4\mv{C}
+2\kappa_k^2\mv{G}
+\mv{G}\mv{C}^{-1}\mv{G}
\right).
\end{equation*}
Here $\mv{C}$ is a diagonal matrix with
$C_{jj} := \int_{\mathcal D}\psi_j(\mv{x})\,d\mv{x}$, and $\mv{G}$ is obtained from
$G_{j\ell} := \int_{\mathcal D}\nabla\psi_j(\mv{x})^\top
\nabla\psi_\ell(\mv{x})\,d\mv{x}$. The local support of the basis functions
makes these matrices, and hence $\mv{Q}_k$, sparse. Thus, the weak FEM
solution is a Gaussian Markov random field that replaces dense
Gaussian-process covariance calculations by sparse matrix operations.

Following the multivariate SPDE construction of
\citet{bolin2020multivariate}, the bivariate field is obtained analogously by
taking the weak solution of the coupled system in \eqref{eq:bispde} and
applying the same Galerkin approximation to both components. Let
$\mv{L} := \operatorname{bdiag}(\mv{L}_1,\mv{L}_2)$,
$\mv{C}_{\mv{S}} := \operatorname{bdiag}(\mv{C},\mv{C})$, and
$\mv{D}_m := \mv{D}(\rho_S)\otimes\mv{I}_m$. The corresponding discrete
operator is $\mv{K} := \mv{D}_m\mv{L}$, and the stacked weights
$\mv{w} := (\mv{w}_1^\top,\mv{w}_2^\top)^\top$ form the joint
Gaussian Markov random field
\begin{equation*}
\mv{w}
\sim
\mathcal N(\mv{0},\mv{Q}^{-1}),
\qquad
\mv{Q}
:=
\mv{K}^\top\mv{C}_{\mv{S}}^{-1}\mv{K}.
\end{equation*}

Finally, let $\mv{A}\in\mathbb R^{n\times m}$ be the observation matrix with
$A_{ij} := \psi_j(\mv{x}_i)$, and define
$\mv{B} := \operatorname{bdiag}(\mv{A},\mv{A})$. The stacked spatial
effects at the observation locations are evaluated as
\begin{equation*}
\begin{pmatrix}
S_1(\mv{x}_1),\ldots,S_1(\mv{x}_n),
S_2(\mv{x}_1),\ldots,S_2(\mv{x}_n)
\end{pmatrix}^\top
\;\approx\;
\mv{B}\mv{w}.
\end{equation*}

\subsection{Inference via the Laplace approximation}
\label{sec:inference}

Maximum likelihood estimation requires integrating out two sets
of latent variables, namely the continuous seroreactivity latent variables $\mv{T}_i$ and the spatial random field $\mv{S}(\mv{x})$. Our proposal is to use the Laplace approximation to
obtain an efficient approximation of the marginal likelihood. Two numerical integrations are required
in the present model. For the latent variables
$\mv{T}_i\in(0,1)^2$ we use the deterministic midpoint
approximation defined in~\eqref{eq:biv_latentT_quadrature}, whereas the much higher-dimensional vector of FEM weights is integrated using a mode-centred Gaussian approximation. We first describe how the Laplace approximation yields the objective function used for parameter estimation, and then give the numerical details of the implementation used for the analyses presented in this paper. 

\begin{algorithm}[H]
\small
\captionsetup{}
\caption{Evaluation  Laplace approximation.}
\label{alg:laplace-evaluation}
\begin{algorithmic}[1]
\Require 
$\vartheta,\mv{w}_{0}$
\State  $k\gets0$
\While{$\|\nabla\varphi(\mv{w}_k)\|_\infty\geq10^{-7}$
and $k<50$} 
\State 
$\mv{W}_\ell^{(k)} :=
\mv{W}_\ell(\mv{w}_k;\vartheta)$. \Comment{Eq.~\eqref{eq:laplace-local-curvature}}
  Build $\mv{H}^{\mathrm{search}}_k$.
  \Comment{Eq.~\eqref{eq:laplace-search-hessian}}
  \State
  $\mv{d}_k\gets
  (\mv{H}^{\mathrm{search}}_k)^{-1}\nabla\varphi(\mv{w}_k)$.
  \State Set $\mv{w}_{k+1}\gets\mv{w}_k-a_k\mv{d}_k$ and
  $k\gets k+1$. \Comment{select $a_k$ by Armijo
  backtracking.}
\EndWhile
\State $\mv{W}_\ell^{(k)} :=
\mv{W}_\ell(\mv{w}_k;\vartheta)$.
\State Build the exact Hessian
$\mv{H}_\vartheta(\widehat{\mv{w}})$.
\Comment{Eq.~\eqref{eq:laplace-exact-hessian}}
\State \Return $\log\widetilde L(\vartheta)$.
\Comment{Eq.~\eqref{eq:laplace}}
\end{algorithmic}
\end{algorithm}

After integrating out
$\mv{T}_i$ on the finite grid, let
$\mv{s}_i := (\mv{B}\mv{w})_i$. The conditional log-likelihood is then
$\ell_\vartheta(\mv{w}) :=
\sum_{i=1}^n\log p(\mv{y}_i\mid\mv{s}_i,
\mv{d}_i;\vartheta)$.

The marginal likelihood for
$\vartheta$ is therefore
\begin{equation*}
L(\vartheta)
:=
\frac{|\mv{Q}|^{1/2}}{(2\pi)^{p/2}}
\int_{\mathbb R^p}
\exp\!\left\{
\ell_\vartheta(\mv{w})
-\tfrac{1}{2}\mv{w}^{\top}\mv{Q}\mv{w}
\right\}
d\mv{w},
\end{equation*}
where $p := 2m$ is the dimension of the bivariate FEM weight
vector. This integral is not available in closed form for
the mixture likelihood used here.

The implemented approximation expands the negative joint
log-density
\begin{equation*}
\varphi_\vartheta(\mv{w})
:=
-\ell_\vartheta(\mv{w})
+\tfrac{1}{2}\mv{w}^{\top}\mv{Q}\mv{w},
\end{equation*}
to second order about a stationary point
$\widehat{\mv{w}}(\vartheta)$.

Let
$\mv{H}_\vartheta(\mv{w}):=
\nabla_{\mv{w}}^2\varphi_\vartheta(\mv{w})$ denote the exact
Hessian of the negative joint log-density. Provided that the mode
search converges to a stationary point $\widehat{\mv{w}}(\vartheta)$
and that $\mv{H}_\vartheta(\widehat{\mv{w}})$ is positive definite,
the leading-order Laplace approximation is the classical
mode-and-curvature construction \citep{tierney1986accurate}:
\begin{equation}
\label{eq:laplace}
\log \widetilde L(\vartheta)
:=
\ell_\vartheta(\widehat{\mv{w}})
-\tfrac{1}{2}\widehat{\mv{w}}^{\top}
\mv{Q}\widehat{\mv{w}}
+\tfrac{1}{2}\log|\mv{Q}|
-\tfrac{1}{2}\log
\left|\mv{H}_\vartheta(\widehat{\mv{w}})\right|.
\end{equation}
No eigenvalue clipping or diagonal jitter is used in the final
determinant. If the exact Hessian is not positive definite, the
corresponding outer-objective evaluation is declared invalid rather
than replacing its curvature by a stabilised approximation.
Equation~\eqref{eq:laplace} is the log-scale form of the mode-evaluated
density-ratio construction in \citep{rueheld2005}.

Algorithm~\ref{alg:laplace-evaluation} gives the computation
performed at each outer objective evaluation. In what follows, we decompose the full parameter vector $\vartheta$ as
$\vartheta := (\vartheta_0,\vartheta_1)$, where $\vartheta_0$
collects the non-spatial parameters (those governing the
observation model~\eqref{eq:obs_joint}, the truncated Gaussian
mixture~\eqref{eq:biv_latentT_density}, and the sequential
logistic weights~\eqref{eq:mixing_probs}) and $\vartheta_1$
collects the spatial covariance parameters
$(\sigma_1,r_1,\sigma_2,r_2,\rho_S)$. The likelihood
and curvature terms entering its conditional-mode search are
obtained as follows. At a given value of the
non-spatial parameter block $\vartheta_0$, the midpoint grid is
used to compute the four component evidences
\begin{equation*}
c_{i\mv{z}}(\vartheta_0)
\approx
\sum_{r=1}^{M}\sum_{s=1}^{M}
\widetilde w_{i,r,s}(\mv{z};\vartheta_0)
p(\mv{y}_i\mid\mv{T}_i=\mv{t}_{r,s};\vartheta_0).
\end{equation*}
Conditional on the projected field
$\mv{s} := \mv{B}\mv{w}$, the individual likelihood then has
the finite-mixture form
\begin{equation*}
p(\mv{y}_i\mid\mv{s}_i,\mv{d}_i;\vartheta)
\approx
\sum_{\mv{z}\in\{0,1\}^2}
\pi_i(\mv{z};\mv{s}_i,\vartheta)
c_{i\mv{z}}(\vartheta_0).
\end{equation*}

The gradient with respect to the projected bivariate field is
computed analytically as a mixture-responsibility-weighted average
of the per-component score contributions, where the responsibility
of component $\mv{z}$ for individual $i$ is the posterior
probability $\pi_i(\mv{z};\mv{s}_i,\vartheta)\,
c_{i\mv{z}}(\vartheta_0) \big/ \sum_{\mv{z}'}
\pi_i(\mv{z}';\mv{s}_i,\vartheta)\, c_{i\mv{z}'}(\vartheta_0)$, and
is aggregated over individuals at the same spatial location. Let
$\mathcal I_\ell$ index the observations at unique location
$\ell$; the resulting raw curvature block is
\begin{equation}
\label{eq:laplace-local-curvature}
\mv{W}_\ell(\mv{w};\vartheta)
:=
-\left.
\nabla^2_{\mv{s}_\ell}
\sum_{r\in\mathcal I_\ell}
\log p(\mv{y}_r\mid\mv{s}_\ell,
\mv{d}_r;\vartheta)
\right|_{\mv{s}_\ell=(\mv{B}\mv{w})_\ell}.
\end{equation}
This is a $2\times2$ matrix because the projected field is
bivariate. Since finite-mixture likelihoods need not be
log-concave, $\mv{W}_\ell$ can be indefinite.

The exact global Hessian of the negative joint
log-density is
\begin{equation}
\label{eq:laplace-exact-hessian}
\mv{H}_\vartheta(\mv{w})
:=
\mv{Q}
+\mv{B}^{\top}
\operatorname{bdiag}\{\mv{W}_1(\mv{w};\vartheta),\ldots,
\mv{W}_L(\mv{w};\vartheta)\}\mv{B}.
\end{equation}
Indefiniteness of an individual likelihood-curvature block does not
by itself invalidate the Laplace approximation, because the precision
term $\mv{Q}$ can make the global Hessian positive definite.

The code
replaces each raw block during mode finding by
\begin{equation*}
\widetilde{\mv{W}}_\ell
:=
\mv{W}_\ell
+\max\{0,-\lambda_{\min}(\mv{W}_\ell)\}\mv{I}_2,
\end{equation*}
so that its smallest eigenvalue is zero.

At iteration $k$, the clipped blocks and any
factorisation jitter are used to construct only the search matrix
\begin{equation}
\label{eq:laplace-search-hessian}
\mv{H}^{\mathrm{search}}_k
:=
\mv{Q}
+\mv{B}^{\top}
\operatorname{bdiag}\{
\widetilde{\mv{W}}_1(\mv{w}_k;\vartheta),\ldots,
\widetilde{\mv{W}}_L(\mv{w}_k;\vartheta)
\}\mv{B}
+j_k\mv{I}_p.
\end{equation}
The value $j_k$ is initially zero and is increased only if the sparse
factorisation fails. This matrix determines the Newton search direction,
whereas the gradient and Armijo line search are evaluated using the
original objective $\varphi_\vartheta$. Thus, neither blockwise clipping
nor numerical jitter changes the mode-defining objective or enters the
Laplace determinant. After convergence, the raw blocks are recomputed at
$\widehat{\mv{w}}$ and assembled into the exact Hessian
in~\eqref{eq:laplace-exact-hessian}.

Sparse Cholesky factorisations are used for the
stabilised search systems and, separately, for the undamped exact Hessian
in the final determinant; both are computed with CHOLMOD
\citep{chen2008cholmod}.
For a quasi-regular two-dimensional
mesh, a factorisation typically requires $O(p^{3/2})$
operations, where $p := 2m$ is the dimension of the bivariate FEM
weight vector $\mv{w}$, compared with $O(p^3)$ for dense
factorisation \citep[Section~2.4.3]{rueheld2005}. If $K$
factorisations are required to locate the conditional mode, a
representative cost per outer evaluation is therefore
$O\{nM^2+(K+1)p^{3/2}\}$, where $n$ is the number of individuals
and $M^2$ is the number of quadrature cells used to approximate
the integral over each $\mv{T}_i$, although the realised cost
depends on the mesh ordering and fill-in.

Let $\eta$ denote the unconstrained working
parameter vector used by the outer optimiser, and write
$\vartheta := \mv{g}(\eta)$ for the transformation to the
natural parameterisation. Positivity constraints are represented
by logarithmic working parameters, correlations by hyperbolic
tangent transformations, and ordered component locations by the
sequential links defined in Section~\ref{sec:model}. The
construction applies to both univariate and bivariate fits. For a
non-spatial fit there are no FEM weights.

For a spatial fit, the Laplace-approximated marginal
log-likelihood is
\begin{equation*}
\widetilde\ell(\eta)
:=
\ell_{\mv{g}(\eta)}
\!\left\{\widehat{\mv{w}}\bigl(\mv{g}(\eta)\bigr)\right\}
-\frac{1}{2}
\widehat{\mv{w}}\bigl(\mv{g}(\eta)\bigr)^\top
\mv{Q}\bigl(\mv{g}(\eta)\bigr)
\widehat{\mv{w}}\bigl(\mv{g}(\eta)\bigr)
+\frac{1}{2}\log\left|\mv{Q}\bigl(\mv{g}(\eta)\bigr)\right|
-\frac{1}{2}\log\left|
\mv{H}_{\mv{g}(\eta)}
\!\left\{\widehat{\mv{w}}\bigl(\mv{g}(\eta)\bigr)\right\}
\right|.
\end{equation*}
This notation makes explicit that both the conditional mode and
the exact determinant are functions of the outer parameters, so
that $\widetilde\ell(\eta)$ is the estimation objective
maximised over the working scale to obtain
$\widehat{\eta}$.

\subsection{Population-level predictive targets}
\label{sec:predictive-targets}

When considering spatial prediction, the hierarchical structure of
the model enables the definition of both population-level and
individual-level predictive targets. In the epidemiological
applications for which the models presented here are intended,
population-level targets are often of greater interest, since
cross-sectional surveys are typically carried out to estimate the
burden of disease or the impact of interventions in the general
population, rather than to characterise any single individual. To
this end, we provide examples of predictive targets that consider
the prediction for a generic individual of a given age $a$ at a
location $\mv{x}$, conditional on the realised spatial field
$\mv{S}(\mv{x}) := \{S_1(\mv{x}),S_2(\mv{x})\}^\top$. Throughout this
section we write $\pi\{\mv{z};a,\mv{S}(\mv{x})\}$ for the
sequential logistic mixing probability
of~\eqref{eq:spatial_mixing_probs}, evaluated at the covariate
vector corresponding to age $a$ and at the spatial field value
$\mv{S}(\mv{x})$. Covariates other than age are suppressed in the
notation for readability, and the individual index $i$ is dropped
because the targets refer to a hypothetical rather than an
observed individual.

A natural choice of predictive target is the expected
seroreactivity for a given antigen $k$,
\begin{equation}
\label{eq:predictive-target-T}
\E\{T_k(a)\mid\mv{S}(\mv{x})\} = \sum_{\mv{z}\in\{0,1\}^2}
\pi\{\mv{z};a,\mv{S}(\mv{x})\}\, \bar t_k(\mv{z}),
\end{equation}
where $\bar t_k(\mv{z}) := \E(T_k\mid\mv{Z}=\mv{z})$ is the $k$th
marginal mean of the truncated Gaussian component
in~\eqref{eq:biv_latentT_component}. Note that $\bar t_k(\mv{z})$
depends on the full vector $\mv{z}$, and not on $z_k$ alone, since
truncation to $(0,1)^2$ couples the two coordinates whenever
$\rho_T\neq0$. This target summarises, at a given location and
age, the expected level of immune activation to antigen $k$ in the
general population, and can be mapped across the study region to
identify areas of higher or lower average exposure, informing the
geographical targeting of interventions or the prioritisation of
surveillance resources towards locations where seroreactivity
indicates ongoing or historically high transmission.

A second predictive target of interest is the expected antibody
concentration for antigen $k$. Recalling
from~\eqref{eq:obs_mean_var} that the mean of $Y_k$ given $T_k=t$
is $\mu_{k,0} + t(\mu_{k,1}-\mu_{k,0})$, the law of total
expectation gives
\begin{equation}
\label{eq:predictive-target-Y}
\E\{Y_k(a)\mid\mv{S}(\mv{x})\} = \sum_{\mv{z}\in\{0,1\}^2}
\pi\{\mv{z};a,\mv{S}(\mv{x})\}\,
\bigl[\mu_{k,0} + \bar t_k(\mv{z})(\mu_{k,1}-\mu_{k,0})\bigr].
\end{equation}
Because this expression is linear in $\bar t_k(\mv{z})$, the
weighted sum can be taken inside the bracket, giving the
equivalent, simpler form
\begin{equation*}
\E\{Y_k(a)\mid\mv{S}(\mv{x})\} = \mu_{k,0} +
\bigl\{\mu_{k,1}-\mu_{k,0}\bigr\}\, \E\{T_k(a)\mid\mv{S}(\mv{x})\},
\end{equation*}
where $\E\{T_k(a)\mid\mv{S}(\mv{x})\}$ is the target
in~\eqref{eq:predictive-target-T}. This simplification relies on
the observation-model mean being an affine function of $T_k$; for
a nonlinear mean it would not hold, and the mixture
sum~\eqref{eq:predictive-target-Y} would need to be evaluated
directly.

These predictive targets are implemented by Monte Carlo
integration over the posterior uncertainty in the spatial field. A
sample of finite-element weights is drawn from the
Laplace-approximate predictive distribution
$\mv{w}^{(b)}\mid\mv{y}\sim
N\!\left\{\widehat{\mv{w}},
\mv{H}_{\widehat\vartheta}(\widehat{\mv{w}})^{-1}\right\}$,
$b=1,\ldots,B$.
Each draw is projected to the location
of interest to give $\mv{S}^{(b)}(\mv{x})$. Because
\eqref{eq:predictive-target-T} and~\eqref{eq:predictive-target-Y}
are already closed-form expressions in $\mv{S}(\mv{x})$, no
further simulation of $\mv{Z}$ or $T_k$ is required at this stage.
Each draw $\mv{S}^{(b)}(\mv{x})$ is substituted directly
into~\eqref{eq:predictive-target-T}
and~\eqref{eq:predictive-target-Y} to obtain $B$ values of the
target, which are then averaged, or summarised by their empirical
distribution to report uncertainty alongside the point estimate,
to give the corresponding Monte Carlo estimate.

The targets in~\eqref{eq:predictive-target-T}
and~\eqref{eq:predictive-target-Y} are age-specific. If interest
instead lies in a summary that is representative of the general
population at a location $\mv{x}$, regardless of age, the effect
of age can be integrated out with respect to the age distribution
$f(a;\mv{x})$ of the population living at $\mv{x}$,
\begin{equation}
\label{eq:predictive-targets-age-marginal}
\E\{T_k\mid\mv{S}(\mv{x})\} = \int
\E\{T_k(a)\mid\mv{S}(\mv{x})\}\, f(a;\mv{x})\, da,
\qquad
\E\{Y_k\mid\mv{S}(\mv{x})\} = \int
\E\{Y_k(a)\mid\mv{S}(\mv{x})\}\, f(a;\mv{x})\, da.
\end{equation}
Monte Carlo approximations to the integrals
in~\eqref{eq:predictive-targets-age-marginal} are easily obtained
by averaging $\E\{T_k(a_j)\mid\mv{S}(\mv{x})\}$ and
$\E\{Y_k(a_j)\mid\mv{S}(\mv{x})\}$ over samples $a_j$ drawn from
$f(a;\mv{x})$. Two approaches can be considered for obtaining such
a sample, depending on the information available. If spatially
varying estimates of the age distribution are available, for
example from gridded population and demographic products such as
those provided by WorldPop \citep{Tatem2017}, or from census or
demographic surveillance data external to the serological survey,
ages can be drawn directly from $f(a;\mv{x})$ at the location of
interest. Alternatively, one can make the simplifying assumption
of a spatially homogeneous age distribution, $f(a;\mv{x})=f(a)$
for all $\mv{x}$, particularly in small study areas where
systematic differences in age structure across the region are
unlikely to be pronounced. Under this assumption, if the
individuals sampled in the serological survey are themselves
representative of the general population of the study area, the
single, population-wide empirical distribution of the observed
ages can be used in place of $f(a)$, so that Monte Carlo samples
are obtained simply by resampling from the ages already present in
the data, without the need for spatially varying external
information.

\section{Application: malaria antibodies in the Kenyan highlands}
\label{sec:application}

We analyse data from a cross-sectional serological survey conducted in Rachuonyo South District, western Kenyan highlands, described in \citet{Bousema2013} and \citet{Bousema2016}. Finger-prick blood samples were collected on filter paper and used to measure total immunoglobulin G antibody responses against two
\textit{Plasmodium falciparum} blood-stage antigens: apical
membrane antigen 1 (AMA1) and merozoite surface protein 1
(MSP1). Antibody concentrations are analysed on the log scale.
After retaining observations with complete data on both antigens
and age, the working dataset contains $15{,}578$ individuals
sampled at $3{,}042$ unique locations. Children under one year
of age are excluded to avoid confounding from maternally derived
antibodies. 

For exploratory analysis, individuals are partitioned
into the age bands $[1,5]$, $(5,10]$, $(10,15]$, $(15,20]$,
$(20,40]$, and $(40,100]$ years. Figure~\ref{fig:ppc-nonspatial} shows the age-stratified empirical distributions of the paired log-AMA1 and log-MSP1 measurements. Two features of the data stand out. First, there is a clear positive association between the two antibody responses that persists across all age groups. This is because both AMA1 and MSP1 are surface antigens expressed on the \textit{Plasmodium falciparum} merozoite during erythrocyte invasion, so a single blood-stage infection event simultaneously stimulates immune responses to both antigens, inducing co-variation in their antibody levels within individuals \citep{Beeson2016, Cowman2017}. Second, this dependence varies with age in ways that are not easily captured by simple parametric models: the joint distribution shifts from a pronounced bimodal
structure at young ages, reflecting a population polarised
between minimal and detectable immune responses, toward a
more concentrated unimodal form at older ages as repeated
exposure accumulates. \citet{giorgiwallin2026} demonstrated the
superiority of the univariate latent variable framework over
standard Gaussian mixture models for capturing this kind of
age-dependent heterogeneity in individual antigen responses.
Here we extend that framework to support the joint modelling
of AMA1 and MSP1, incorporating both age-dependent cross-antigen
dependence and the spatial correlation that arises from the
clustering of malaria transmission risk at fine geographical
scales.

We fit the bivariate geostatistical latent variable model of
Sections~\ref{sec:model} and~\ref{sec:spatial}, with a
specification motivated by the exploratory patterns above. The
covariate vector $\mv{d}_i$ entering the component-location
predictors \eqref{eq:biv_latentT_location_links} and the mixing
predictors \eqref{eq:mixing_probs} is a linear spline in
$\log(a_i)$ with a knot at age $10$ years,
\begin{equation}
\label{eq:age-spline}
\mv{d}_i := \bigl(1,\; \log a_i,\; \{\log a_i - \log 10\}_+\bigr)^\top,
\end{equation}
where $\{u\}_+ := \max(u,0)$, allowing the rate of seroreactivity
acquisition to differ between young children and older
individuals, consistent with the bimodal-to-unimodal transition
seen in Figure~\ref{fig:ppc-nonspatial}. We use a change point in
the age effect because, as shown in \citet{giorgiwallin2026}, the age dependence of seroreactivity for both antigens
is not well described by a single smooth trend across the whole age range. We set the knot at $10$ years based on empirical evidence, as this
choice provided a satisfactory fit based on the validation approach described below.

The cross-antigen association retains the exponential form of
\eqref{eq:mixing_probs},
\begin{equation*}
\delta(a_i) := \exp\bigl(\delta_0 + \delta_1 \log(a_i)\bigr),
\end{equation*}
entering, by the model's sequential construction, only the
conditional probability of high MSP1 seroreactivity given the
AMA1 regime. We denote the spatial scale, defined as the practical range
$\omega_k := \sqrt{8\nu_{k}}/\kappa_k$ of Section~\ref{sec:spatial}, by
$\omega_1$ and $\omega_2$ for AMA1 and MSP1 respectively. All parameters are
estimated using the
Laplace-approximated log-likelihood \eqref{eq:laplace}, and
uncertainty is quantified by a parametric bootstrap with
$1{,}500$ replicates \citep{efron1979bootstrap}.

Table~\ref{tab:bootstrap_redhot} reports the resulting estimates
and 95\% bootstrap confidence intervals. The estimated spatial scales, $\omega_1 \approx 473$m for AMA1
and $\omega_2 \approx 437$m  for MSP1, indicate strong,
fine-scale spatial variation in the probability of high seroreactivity for
both antigens, consistent
with the residential-scale clustering of malaria transmission
previously documented in this same study district using these
same two serological markers \citep{Stresman2017}. The estimated
cross-field correlation, $\rho_S = 0.587$, further indicates substantial spatial co-localisation between the two antigen-specific
probabilities of high seroreactivity, reflecting shared
environmental determinants of exposure to \textit{Plasmodium
falciparum}. The within-component correlation $\rho_T = 0.717$  indicates that, conditional on the exposure
environment at a given location $\mathbf{x}$, individuals with
higher AMA1 seroreactivity also tend to exhibit higher MSP1
seroreactivity, reflecting shared within-host immunological
mechanisms rather than shared spatial exposure alone. The
positive estimate $\delta_1 = 0.286$ indicates that this
cross-antigen association strengthens with age, consistent with
the accumulation of co-stimulation through repeated exposures
over the life course. For AMA1, the estimated standard deviation at the high-seroreactivity
boundary ($\sigma_{1,1}$) is roughly half that at the low-seroreactivity
boundary ($\sigma_{1,0}$); for MSP1 the reduction is more modest. Both
patterns are consistent with the expected assay structure, with
low- and high-seroreactivity means well separated for both antigens and
reduced spread at high seroreactivity as a result of assay saturation.

\begin{table}[ht!]
\centering
\footnotesize
\caption{Point estimates and 95\% bootstrap
confidence intervals for the parameters of the joint bivariate spatial model,
based on 1{,}500 parametric bootstrap replicates. For
antigen $k$, $\alpha_{k,0,1}$, $\alpha_{k,0,2}$, and $\alpha_{k,0,3}$ denote,
respectively, the intercept, the log-age slope, and the additional
change in slope above age 10 years, of the predictor $m_k(\mv{d}_i;0)$ for the
location of the low-seroreactivity component; $\alpha_{k,1,1}$, $\alpha_{k,1,2}$,
and $\alpha_{k,1,3}$ denote the same three coefficients of the predictor
$m_k(\mv{d}_i;1)$ for the location of the high-seroreactivity component; and
$\gamma_{k,1}$, $\gamma_{k,2}$, and $\gamma_{k,3}$ denote the same three
coefficients of the predictor for the mixing probability of antigen $k$, all
defined through the linear age spline in~\eqref{eq:age-spline}. The spatial
scale $\omega_k := \sqrt{8\nu_{k}}/\kappa_k$ is the practical range of
Section~\ref{sec:biv-matern}, with $\nu_{k}=1$ fixed for both antigens; $\zeta_k$
is the scale parameter of the truncated Gaussian mixture in
Eq.~\eqref{eq:biv_latentT_covariance}.}
\label{tab:bootstrap_redhot}
\begin{tabular}{llrrrrrr}
\toprule
 & & \multicolumn{3}{c}{AMA1 ($k=1$)} & \multicolumn{3}{c}{MSP1 ($k=2$)} \\
\cmidrule(lr){3-5} \cmidrule(lr){6-8}
Process & Parameter & Estimate & 2.5\% & 97.5\% & Estimate & 2.5\% & 97.5\% \\
\hline
\multirow{3}{*}{Spatial process $S$}
 & $\omega_k$ & 473.354 & 352.511 & 638.694 & 437.042 & 317.715 & 577.241 \\
 & $\sigma_k$ & 0.962 & 0.827 & 1.132 & 0.790 & 0.668 & 0.934 \\
 & $\rho_S$ & \multicolumn{6}{c}{0.587 (0.453, 0.711)} \\
\hline
\multirow{17}{*}{Seroreactivity process $T$}
 & \multicolumn{7}{l}{\textit{Low seroreactivity component}} \\
 & $\alpha_{k,0,1}$ & 0.296 & 0.256 & 0.327 & 0.416 & 0.360 & 0.463 \\
 & $\alpha_{k,0,2}$ & 0.104 & 0.096 & 0.113 & 0.049 & 0.042 & 0.056 \\
 & $\alpha_{k,0,3}$ & -0.038 & -0.051 & -0.025 & -0.001 & -0.012 & 0.012 \\
 & \multicolumn{7}{l}{\textit{High seroreactivity component}} \\
 & $\alpha_{k,1,1}$ & -0.733 & -0.829 & -0.642 & -1.209 & -1.275 & -1.133 \\
 & $\alpha_{k,1,2}$ & -0.179 & -0.220 & -0.131 & -0.087 & -0.114 & -0.062 \\
 & $\alpha_{k,1,3}$ & -0.342 & -0.424 & -0.264 & 0.029 & -0.026 & 0.085 \\
 & \multicolumn{7}{l}{\textit{Mixing probabilities}} \\
 & $\gamma_{k,1}$ & -3.889 & -4.282 & -3.412 & -0.465 & -0.738 & -0.190 \\
 & $\gamma_{k,2}$ & 1.800 & 1.617 & 1.959 & -0.723 & -0.818 & -0.635 \\
 & $\gamma_{k,3}$ & -1.667 & -1.879 & -1.380 & 0.508 & 0.297 & 0.779 \\
 & $\delta_0$ & -- & -- & -- & 0.124 & -0.024 & 0.328 \\
 & $\delta_1$ & -- & -- & -- & 0.286 & 0.211 & 0.339 \\
 & \multicolumn{7}{l}{\textit{Covariance parameters}} \\
 & $\zeta_k$ & 0.127 & 0.120 & 0.134 & 0.122 & 0.112 & 0.134 \\
 & $\rho_T$ & \multicolumn{6}{c}{0.717 (0.691, 0.750)} \\
\hline
\multirow{4}{*}{Antibody measurement $Y$}
 & $\mu_{k,0}$ & -5.613 & -5.885 & -5.312 & -6.061 & -6.602 & -5.535 \\
 & $\mu_{k,1}$ & 0.824 & 0.781 & 0.867 & 0.862 & 0.786 & 0.942 \\
 & $\sigma_{k,0}$ & 0.458 & 0.390 & 0.539 & 0.386 & 0.286 & 0.495 \\
 & $\sigma_{k,1}$ & 0.226 & 0.202 & 0.252 & 0.273 & 0.237 & 0.315 \\
\bottomrule
\end{tabular}
\end{table}

We assess the adequacy of the model using the graphical validation approach adapted from \citet{giorgiwallin2026} (Section~4.2), comparing
the observed age-stratified joint and marginal distributions of
log-AMA1 and log-MSP1 with the corresponding distributions
simulated from the fitted model, preserving the observed age
structure of the sample. The right panel of Figure~\ref{fig:ppc-nonspatial} shows the results of this comparison. The simulated predictive histograms adequately approximate the observed data within every age band, with small discrepancies observed for the age group $(5, 10]$ for both AMA1 and MSP1. Overall, this indicates that the fitted model provides a satisfactory representation of the observed age-dependent distributions of AMA1 and MSP1 concentrations.

Figure~\ref{fig:application-prediction} summarises the
model-based predictions obtained by projecting the fitted
spatial field to the survey locations and forward-simulating the
component indicators, latent seroreactivity, and antibody
concentration from the fitted parameters. Panel (a) shows a
clear spatially structured pattern in the predicted probability
of high seroreactivity, $p_k(\mathbf{x})$, for both antigens,
consistent with the estimated spatial scales
and the substantial cross-field correlation reported in
Table~\ref{tab:bootstrap_redhot}. This spatial structure is
visibly attenuated in panel (b), which shows the corresponding
predicted latent seroreactivity $t_{i,k}(\mathbf{x})$, as this introduces additional non-spatial variation due to individual-level biological factors. Panels (c) and (d) indicate that both
higher predicted $p_k$ and higher predicted $t_{i,k}$ correspond to
higher observed antibody concentrations, as expected. The two
panels differ, however, in the shape of this relationship. For
MSP1, the observed antibody distribution remains largely flat across the first three quintiles of $p_k$ and only shifts markedly between the fourth and fifth quintiles, whereas the corresponding increase across quintiles of $t_{i,k}$ is close to linear. 

\begin{figure}[hp]
\centering
\includegraphics[width=\textwidth,height=0.88\textheight,keepaspectratio]{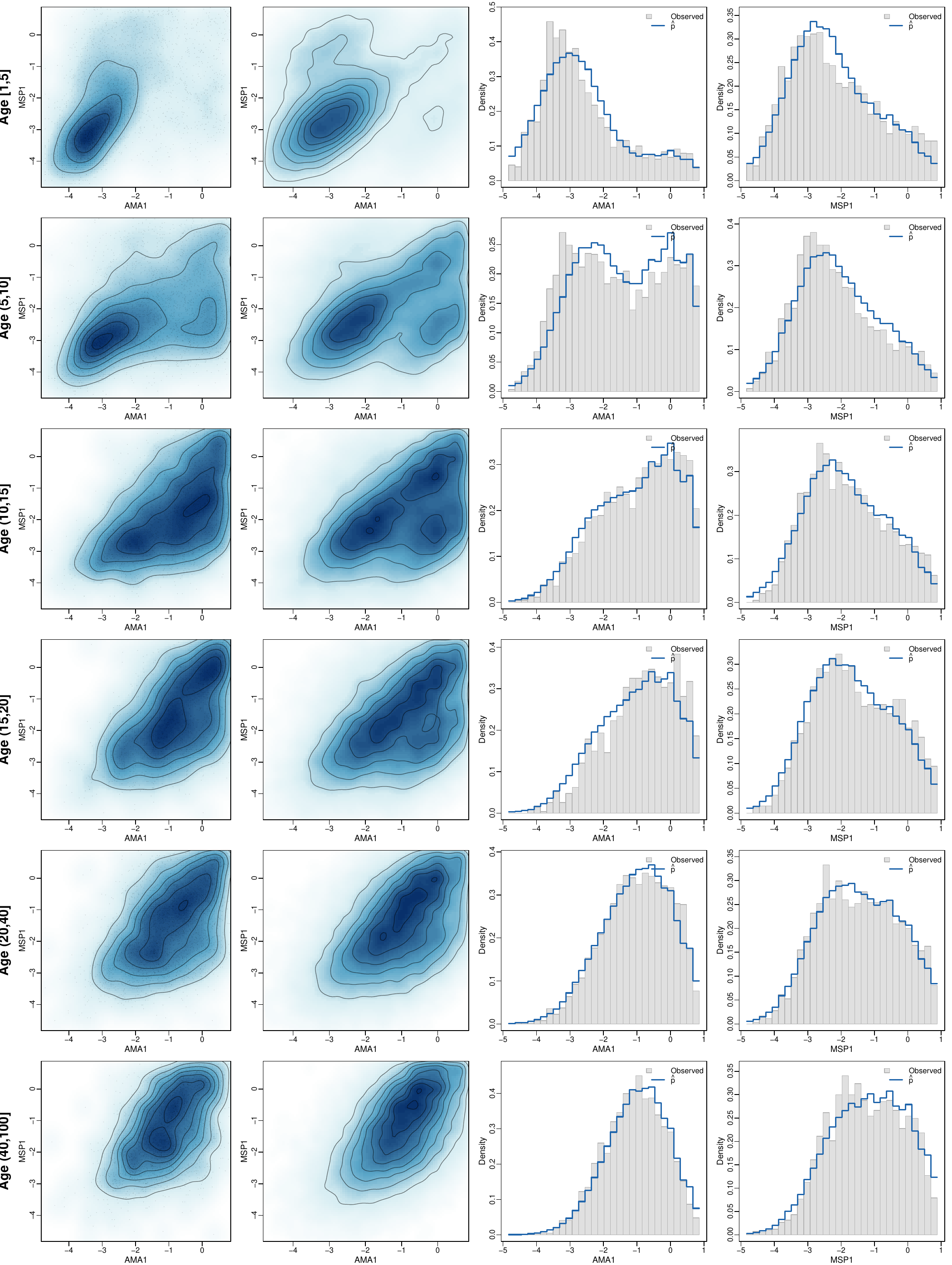}
\caption{Age-stratified predictive check for
the bivariate spatial model. Rows correspond to age
bands. The first column shows the observed joint
distribution of log-AMA1 and log-MSP1, the second shows
the simulated predictive joint density, and the two right
columns compare observed marginal histograms with
simulated marginal densities for AMA1 and MSP1.}
\label{fig:ppc-nonspatial}
\end{figure}

\begin{figure}[hp]
\centering
\includegraphics[width=\textwidth]{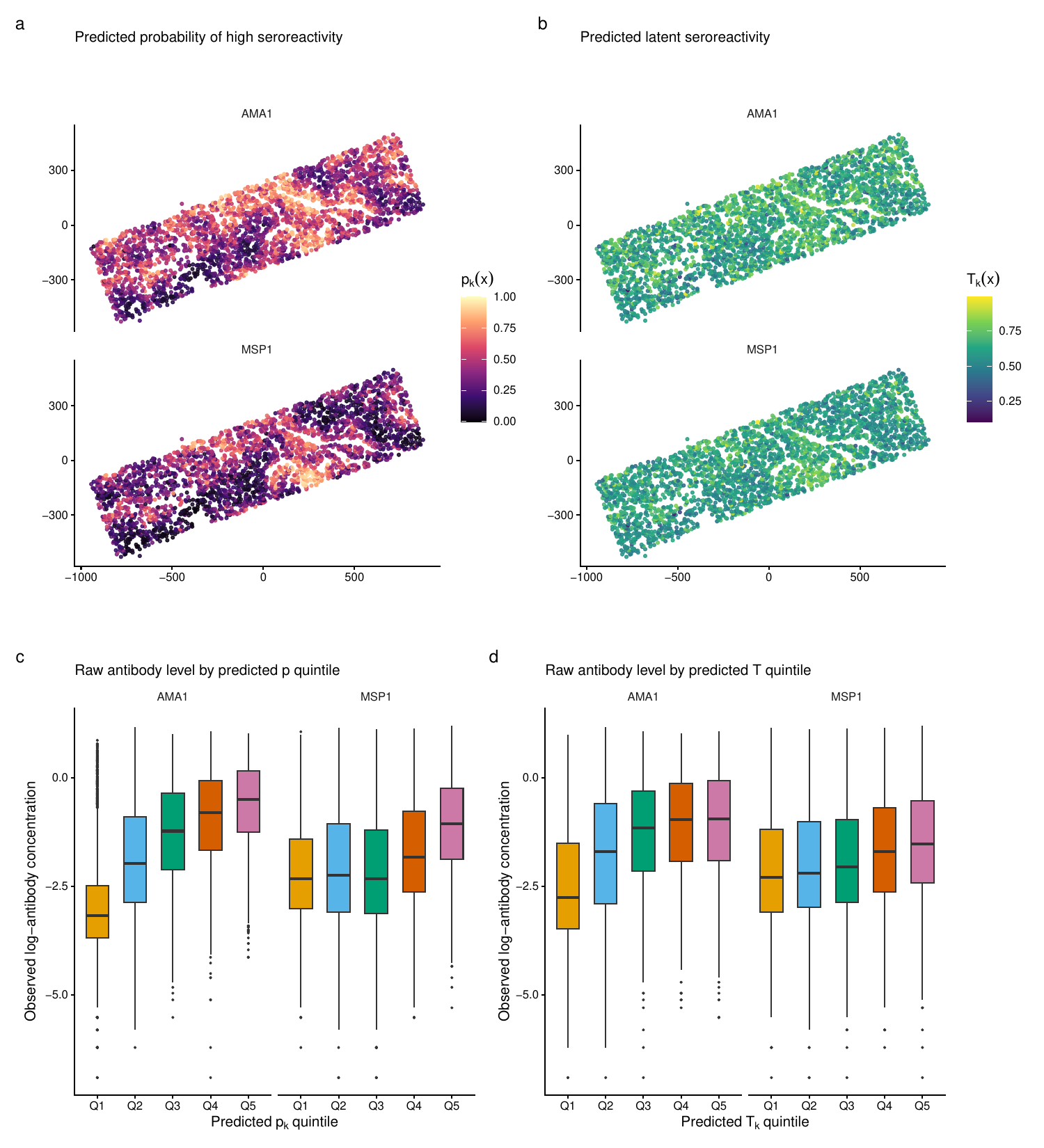}
\caption{Model-based predictions from the fitted joint bivariate spatial
model, obtained by drawing one sample from the Laplace-approximate
posterior of the spatial field, projecting it to the survey locations, and
forward-simulating the component indicators, latent seroreactivity, and
antibody concentration from the fitted parameters. Panel (a): predicted
probability of high seroreactivity, $p_k(\mathbf{x})$, averaged over
residents at each of the $3{,}042$ unique locations. Panel (b): predicted
latent seroreactivity, $t_{i,k}(\mathbf{x})$, similarly averaged by location.
Panels (c) and (d): boxplots of the raw observed log-antibody concentration
$y_k$ across quintiles of the predicted $p_k$ and $t_{i,k}$ respectively,
computed separately for each antigen; quintile labels denote only the
ordinal ranking (Q1 = lowest to Q5 = highest), not the underlying numeric
class boundaries. All panels are shown separately for AMA1 and MSP1.}
\label{fig:application-prediction}
\end{figure}

\subsection{Assessment of the recovery of the joint antibody distribution}
\label{sec:tv}

\begin{table}[ht!]
\centering
\footnotesize
\caption{Joint coarsened total variation distance $\widehat
d_{\mathrm{TV},\mathcal{C}}$ between the observed and model-predicted bivariate
distribution of log-AMA1 and log-MSP1, for the univariate and joint models
with and without the spatial field. All values use a common non-uniform
$10\times10$ AMA1--MSP1 grid pooled from pilot draws under all four fitted
models. The final row compares the full empirical distribution with
predictions at randomly sampled cohort covariates. Lower values indicate closer
agreement.}
\label{tab:tv_joint}
\begin{tabular}{lcccc}
\toprule
 & \multicolumn{2}{c}{\textbf{Non-spatial}} & \multicolumn{2}{c}{\textbf{Spatial}} \\
\cmidrule(lr){2-3} \cmidrule(lr){4-5}
\textbf{Age band} & Univ. & Joint & Univ. & Joint \\
\midrule
$[1,5]$     & 0.295 & 0.189 & 0.278 & 0.192 \\
$(5,10]$    & 0.204 & 0.129 & 0.186 & 0.137 \\
$(10,15]$   & 0.233 & 0.125 & 0.208 & 0.125 \\
$(15,20]$   & 0.254 & 0.131 & 0.234 & 0.133 \\
$(20,40]$   & 0.225 & 0.083 & 0.211 & 0.083 \\
$(40,100]$  & 0.231 & 0.106 & 0.226 & 0.105 \\
\addlinespace
All ages    & 0.194 & 0.078 & 0.175 & 0.079 \\
\bottomrule
\end{tabular}
\end{table}

\begin{table}[ht!]
\centering
\footnotesize
\caption{All-age marginal coarsened total variation distance $\widehat
d_{\mathrm{TV},\mathcal{C}}$ for each antigen separately, on a common $30$-cell grid
per antigen. }
\label{tab:tv_marginal}
\begin{tabular}{lrrrr}
\toprule
 & \multicolumn{2}{c}{\textbf{Non-spatial}} & \multicolumn{2}{c}{\textbf{Spatial}} \\
\cmidrule(lr){2-3} \cmidrule(lr){4-5}
\textbf{Antigen} & Joint & Univ. & Joint & Univ. \\
\midrule
AMA1 & 0.043 & 0.026 & 0.042 & 0.031 \\
MSP1 & 0.035 & 0.028 & 0.035 & 0.031 \\
\bottomrule
\end{tabular}
\end{table}

Here, we assess how well the fitted models reproduce the observed
bivariate distribution of $\mv{Y} := (Y_1,Y_2)^\top$. We compare four fitted
models, obtained by crossing two model choices: whether cross-antigen
dependence is included, giving the dependent model of
Sections~\ref{sec:model}--\ref{sec:spatial}, or univariate separate models with
$\rho_T=\rho_S=0$; and whether the spatial field is included, giving the
spatial or non-spatial specification. For each of the four resulting
models, we compare the observed and predicted distributions separately for
the joint pair $(Y_1,Y_2)$ and for each antigen's marginal, so that any
improvement can be attributed to the right source. As a summary
of the discrepancy between an observed and a predicted distribution we use
the total variation (TV) distance. For two absolutely continuous densities $f$ and
$g$, the TV distance is one half of their $L_1$ distance 
\citep{scheffe1947,gibbs2002},
\begin{equation}
\label{eq:tv}
d_{\mathrm{TV}}(f,g) := \frac{1}{2}\int |f(\mv{y})-g(\mv{y})|\,d\mv{y}
\;\in\;[0,1],
\end{equation}
which can be interpreted as the proportion of probability mass that is not shared by the two distributions.

A direct empirical analogue of \eqref{eq:tv} is degenerate for continuous
responses. Indeed, the empirical measure is supported on finitely many
observed points, so its total variation distance from an absolutely
continuous fitted distribution is one. We therefore compare the
distributions after coarsening the response space to a finite partition
$C_1,\ldots,C_J$. For densities $f$ and $g$, let
$p_j := \int_{C_j}f(\mv y)\,d\mv y$ and
$q_j := \int_{C_j}g(\mv y)\,d\mv y$ denote the corresponding cell
probabilities. The coarsened total variation distance satisfies
\[
    d_{\mathrm{TV},\mathcal C}(f,g)
    :=
    \frac{1}{2}\sum_{j=1}^{J}|p_j-q_j|
    \leq d_{\mathrm{TV}}(f,g),
\]
since coarsening cannot increase total variation \citep{csiszar1967}.
Thus, it gives a partition-dependent lower bound on \eqref{eq:tv}.

Since the true data-generating distribution $f$ is unknown, we replace
$p_j$ by its empirical estimate
$\widehat p_j := n^{-1}\sum_{i=1}^n\mathbb{I}(\mv y_i\in C_j)$.
Similarly, we estimate $q_j$ by the corresponding cell frequency
$\widehat q_j$ from $M_{\mathrm{TV}}$ Monte Carlo draws from the fitted
model. Our reported statistic is therefore
\begin{equation*}
    \widehat d_{\mathrm{TV},\mathcal C}
    :=
    \frac{1}{2}\sum_{j=1}^{J}
    \bigl|\widehat p_j-\widehat q_j\bigr|.
\end{equation*}

We use the same partition $\mathcal C := \{C_1,\ldots,C_J\}$ for all models
within each comparison, with cell boundaries determined by fitted-model
quantiles pooled across the four models. This gives finer resolution in
regions of high predictive density, while the outer cells extend to
$\pm\infty$. For the joint comparison, $\mathcal C$ is a $10\times10$
product grid, whereas each marginal comparison uses 30 cells per antigen.

For age-specific comparisons, we simulate one response for each individual
in the corresponding age band, conditional on that individual's age and,
for the spatial models, fitted spatial field value. For the comparison
across all ages, each Monte Carlo draw is obtained by sampling an individual
from the cohort with replacement and simulating a response conditional on
that individual's age and covariates (if used). Thus, the simulated sample reflects
the empirical age and covariate distribution of the cohort, and
$\widehat q_j$ is obtained as the resulting proportion of simulated
responses falling in $C_j$.

Table~\ref{tab:tv_joint} reports $\widehat d_{\mathrm{TV},\mathcal{C}}$
for the joint pair $(Y_1,Y_2)$. In every age band, and under both the
spatial and non-spatial specifications, this distance is substantially
smaller for the joint model than for the univariate one. For example,
across all ages it falls from $0.194$ to $0.078$ without the spatial field,
and from $0.175$ to $0.079$ with it, differences an order of magnitude
larger than the Monte Carlo standard errors. Recalling that
$\widehat d_{\mathrm{TV}}$ is the proportion of probability mass placed in
the wrong grid cell, this means that around $19\%$ of the joint
distribution is misallocated by the univariate model, against around $8\%$
for the joint model. Table~\ref{tab:tv_marginal} shows that the univariate models are, however, slightly better at capturing each antigen's marginal distribution than the joint model, which is expected since they are optimised to capture the marginal distributional shape only. 

These results indicate that the joint model yields a substantial improvement in recovering the bivariate distribution of the antibody outcomes. This gain comes at the cost of a slightly poorer fit to the marginal distributions compared with univariate models that ignore dependence.

\subsection{Assessing predictive performance via cross-validation}
\label{sec:cv_comparison}

We now examine predictive performance for the same four models introduced in
Section~\ref{sec:tv}. Let $\mv{x}_1,\ldots,\mv{x}_{n_{L}}$ denote the unique locations of
Section~\ref{sec:application}. We partition these locations at random into
$V=5$ folds of approximately equal size,
$\mathcal{L}_1,\ldots,\mathcal{L}_V$, and hold out one fold at a time.
Individuals residing at a location in $\mathcal{L}_v$ form the test set for
fold $v$ and all remaining individuals form the training set. Within each fold the
regression, mixture and spatial covariance parameters are held fixed at the
estimates obtained once from the complete data, and only the latent field mode $\widehat{\mv{w}}_{-v}$ and exact Hessian
$\mv{H}_{\widehat\vartheta,-v}(\widehat{\mv{w}}_{-v})$ are recomputed
from the training locations. Hence,
after using all available data to recover the model parameters, cross-validation
is used only to assess the predictive contribution of the stochastic
components.

Predictive accuracy is measured by three proper scoring rules
\citep{gneiting2007strictly}. For antigen $k$ with marginal predictive CDF
$F_k$, the continuous ranked probability score is
\begin{equation}
\label{eq:crps}
\mathrm{CRPS}(F_k,y_k) := \int_{-\infty}^{\infty}
\bigl\{F_k(u) - \mathbb{I}(u\geq y_k)\bigr\}^2\, du = \E|Y_k-y_k|
   - \frac{1}{2}\E|Y_k-Y_k'|,
\end{equation}
where $Y_k,Y_k'\stackrel{\mathrm{iid}}{\sim}F_k$. The rule used by \citet{gneiting2010matern} to compare independent against
bivariate Mat\'ern co-kriging, so our comparison is directly analogous. Since, the
CRPS is not locally scale invariant and gives more weight to observations
whose predictive distribution has large uncertainty
\citep{bolinwallin2023}, we also report the scaled
CRPS,
\begin{equation}
\label{eq:scrps}
\mathrm{SCRPS}(F_k,y_k) :=
\frac{\E|Y_k-y_k|}{\E|Y_k-Y_k'|}
+ \tfrac{1}{2}\log \E|Y_k-Y_k'|,
\end{equation}
where $Y_k,Y_k'\stackrel{\mathrm{iid}}{\sim}F_k$, which is locally scale
invariant and negatively oriented as in~\eqref{eq:crps}. Both
\eqref{eq:crps} and \eqref{eq:scrps} are univariate and reward cross-antigen
dependence only indirectly, through the marginals. We therefore also report
the energy score, the multivariate generalisation of~\eqref{eq:crps},
\begin{equation}
\label{eq:es}
\mathrm{ES}(F,\mv{y}) := \E\|\mv{Y}-\mv{y}\|
- \tfrac{1}{2}\E\|\mv{Y}-\mv{Y}'\|,
\end{equation}
where $\mv{Y},\mv{Y}'\stackrel{\mathrm{iid}}{\sim}F$, the bivariate
predictive distribution of the pair, and $\|\cdot\|$ is the Euclidean norm.
Unlike \eqref{eq:crps} and \eqref{eq:scrps}, equation~\eqref{eq:es} is
minimised in expectation only by the correctly specified joint distribution.
For all three scores, lower values indicate better predictive performance.
All are computed directly from the Monte Carlo draws using their empirical
estimators. For the second expectation in the energy score,
each predictive draw is paired cyclically with the following independent draw.
This gives an unbiased estimator of $\E\|\mv{Y}-\mv{Y}'\|$ whose computational
cost is linear, rather than quadratic, in the number of predictive draws.

We complement the scores with two calibration and sharpness summaries. The
first is the empirical coverage and mean width of equal-tailed $95\%$
prediction intervals for each antigen. The second takes into account the joint
predictive distribution explicitly. For a held-out individual with predictive density
$f_i$ and target probability $p=0.95$, the highest-density predictive region
is
\begin{equation}
\label{eq:hdr}
R_{i,p} := \{\mv{y}: f_i(\mv{y}) \geq c_{i,p}\},
\end{equation}
where $c_{i,p}$ is the largest threshold satisfying
$\Pr(\mv{Y}_i\in R_{i,p})\geq p$ \citep{hyndman1996}. Among all regions with
the required probability content, \eqref{eq:hdr} has minimum area, and it
may be disconnected when the predictive distribution is multimodal. At
comparable coverage, the smaller region is the sharper prediction. Because the
predictive distribution has no closed form, $f_i$ is estimated by a
diagonal-bandwidth Gaussian kernel density estimate \citep{silverman1986}.
For the reported HDR comparisons, we use $M_{\mathrm{CV}}=3000$
paired predictive draws per individual, split equally between estimation of
the kernel density and calibration of the threshold $c_{i,p}$. The estimated density is evaluated
on an $80\times80$ midpoint grid, and the area of $R_{i,p}$ is the total area
of grid cells whose estimated density exceeds $c_{i,p}$. We average the scores
at the location level for each model.

\begin{table}[ht!]
\centering
\footnotesize
\caption{Parameter-fixed five-fold spatial cross-validation results for the malaria data, based on $15{,}578$ individuals at $3{,}042$ unique locations. Results are reported for four models: joint spatial, univariate spatial, joint non-spatial, and univariate non-spatial. Metrics reported are the continuous ranked probability score (CRPS), scaled CRPS (SCRPS), energy score (ES), marginal $95\%$ prediction interval coverage and mean width for each antigen, and bivariate $95\%$ highest-density region (HDR) coverage, mean area, and median area for the antibody pair. Scoring rules and prediction intervals are computed using $M_{\mathrm{CV}}=3000$ predictive draws per individual.}
\label{tab:cv_redhot}
\begin{tabular}{llrrrr}
\toprule
 & & \multicolumn{2}{c}{\textbf{Spatial}} & \multicolumn{2}{c}{\textbf{Non-spatial}} \\
\cmidrule(lr){3-4} \cmidrule(lr){5-6}
\textbf{Metric} & \textbf{Outcome} & Joint & Univariate & Joint & Univariate \\
\midrule
\multicolumn{6}{l}{\textit{Proper scoring rules}} \\
CRPS  & AMA1 & 0.625 & 0.613 & 0.661 & 0.653 \\
CRPS  & MSP1 & 0.677 & 0.674 & 0.715 & 0.711 \\
SCRPS & AMA1 & 1.106 & 1.098 & 1.134 & 1.127 \\
SCRPS & MSP1 & 1.150 & 1.148 & 1.178 & 1.175 \\
ES    & Pair & 1.009 & 1.007 & 1.064 & 1.067 \\
\addlinespace
\multicolumn{6}{l}{\textit{Marginal $95\%$ prediction intervals}} \\
Coverage (\%) & AMA1 & 95.603 & 96.020 & 95.031 & 95.609 \\
Coverage (\%) & MSP1 & 94.903 & 95.442 & 94.935 & 95.218 \\
Mean width    & AMA1 & 4.204 & 4.273 & 4.260 & 4.356 \\
Mean width    & MSP1 & 4.490 & 4.529 & 4.587 & 4.631 \\
\addlinespace
\multicolumn{6}{l}{\textit{Bivariate $95\%$ highest-density region}} \\
Coverage (\%) & Pair & 95.166 & 95.019 & 95.121 & 94.800 \\
Mean area     & Pair & 19.734 & 22.579 & 20.323 & 23.442 \\
Median area   & Pair & 19.286 & 21.948 & 20.112 & 22.950 \\
\bottomrule
\end{tabular}
\end{table}

Table~\ref{tab:cv_redhot} reports the cross-validation results. The clearest
gain in predictive accuracy comes from introducing the spatial field. For
both the joint and the univariate specifications, the spatial models have
lower energy scores than their non-spatial counterparts, with paired
differences of $-0.055$ and $-0.060$, respectively. The associated standard
errors indicate that these differences are unlikely to arise from sampling
variability alone. For example, the corresponding standard errors for the differences reported in Table~\ref{tab:cv_redhot} are, in fact, 0.0026 and 0.0029 (these standard errors are computed by treating unique locations as the clustering units and assuming independence across locations). The spatial models also
consistently have lower CRPS and SCRPS for both antigens. These results
indicate a consistent gain in predictive accuracy from accounting for
residual spatial variation in seroreactivity when predicting at unsampled locations.

The main advantage of the bivariate model is instead seen in the joint predictive regions. The joint spatial
model attains $95.17\%$ coverage with a mean HDR area of $19.73$, whereas the
separate univariate spatial models attain $95.02\%$ coverage with a mean area
of $22.58$. The paired mean-area difference is $-2.85$, corresponding to a
reduction of $12.60\%$ at essentially the same coverage. The corresponding
non-spatial comparison gives $95.12\%$ coverage and a mean area of $20.32$
for the joint model, against $94.80\%$ coverage and a mean area of $23.44$
for the separate univariate models. The paired mean-area difference is
$-3.12$, corresponding to a reduction of $13.31\%$. While the spatial field
provides the main gain in predictive accuracy, modelling the two antigens
jointly provides a substantial gain in the sharpness of the bivariate
predictive distribution. The joint models also yield slightly narrower marginal
prediction intervals while retaining coverage close to the nominal level.

\section{Simulation study}
\label{sec:simulation}

The objective of this simulation study is to quantify the benefit of jointly
modelling both the cross-antigen dependence in the latent seroreactivity
process and the spatial cross-correlation between antigen-specific random
fields, relative to fitting a separate spatial model to each antigen that
ignores both sources of dependence. We focus on the performance of predictive
inferences for three quantities of primary interest: the spatial random field
$S$, the latent seroreactivity process $T$ and the observed antibody
outcome $Y$.

Using the data described in Section~\ref{sec:application}, we simulate data
over a randomly selected subset of $n_L$ survey locations, using
the maximum likelihood estimates from the joint model of
Section~\ref{sec:application} as the true parameter values. We simulate five individuals per location,
approximating the average number of individuals per location in the
application of Section~\ref{sec:application}. Two correlation parameters are
varied over a $3 \times 3$ grid of values: the spatial cross-field coupling
$\rho_S \in \{0.1, 0.5, 0.9\}$, and the within-component latent correlation
$\rho_T \in \{0.1, 0.5, 0.9\}$, with all other parameters, including the
unequal scales $\omega_1$ and $\omega_2$, held at their fitted values. We
additionally vary the amount of measurement noise in the observation model
through a variance ratio $c \in \{1, 0.5, 0.25, 0.1\}$, which multiplies the
fitted conditional variances of $Y$ given $T$, so that $c=1$ reproduces the
noise level estimated in the application and smaller values correspond to a
more precise assay. 
For each of the resulting 36 scenarios, data are
simulated at three values for the number of sampled locations,
$n_L \in \{300, 600, 1200\}$.

For every simulated dataset we fit three models: the joint bivariate
spatial model of Section~\ref{sec:spatial}, and two separate univariate
spatial models, one per antigen, which ignore the cross-antibody
association and treat each spatial field as independent. From each fitted
model we predict the spatial field at every sampled location, and use these
predictions together to forward-simulate and predict the corresponding latent
seroreactivity and antibody concentration. 
This yields, for every sampled location $\mathbf{x}_i$,
$i=1,\ldots,n_L$, antigen $k=1,2$, and fitted model,
predicted values $\widehat S_k(\mathbf{x}_i)$, $\widehat t_{i,k}(\mathbf{x}_i)$
and $\widehat Y_k(\mathbf{x}_i)$ that can be compared directly with the
corresponding simulated true values $S_k(\mathbf{x}_i)$,
$t_{i,k}(\mathbf{x}_i)$ and $Y_k(\mathbf{x}_i)$.
Consistent with the population-level targets of
Section~\ref{sec:predictive-targets}, these are predictions for
generic individuals of a given age at sampled locations,
rather than subject-specific conditional estimates on the
corresponding antibody outcomes.

Predictive performance for $S$ is summarised, pooling both antigens
together, by the median bias and median absolute error,
\begin{align*}
\text{MB}(S) &:= \operatorname{median}_{k,i}
\left\{\widehat S_k(\mathbf{x}_i) - S_k(\mathbf{x}_i)\right\}, \\
\text{MAE}(S) &:= \operatorname{median}_{k,i}
\left|\widehat S_k(\mathbf{x}_i) - S_k(\mathbf{x}_i)\right|,
\end{align*}
computed within each $(n_L, c, \rho_S, \rho_T)$ scenario over all replicates.
The same expressions, with $S$ replaced by $T$ and by $Y$, are used to
summarise predictive performance for the latent seroreactivity and the
observed antibody concentration. 

\subsection{Results}
\begin{table}[ht!]
\centering
\footnotesize
\caption{Median bias (MB) and median absolute error (MAE) of $S$, $T$ and $Y$ (pooled across both antigens), comparing the separate univariate models (Univ.) and the joint bivariate model (Joint), across the 9 $(\rho_S,\rho_T)$ scenarios and the 4 variance-ratio values $c$, sample size $n_{L} = 300$.}
\label{tab:bias_mae_redhot_n300_wide}
\resizebox{\textwidth}{!}{%
\begin{tabular}{lllrrrrrrrrrrrr}
\toprule
$c$ & $\rho_S$ & $\rho_T$ & \multicolumn{4}{c}{$S$} & \multicolumn{4}{c}{$T$} & \multicolumn{4}{c}{$Y$} \\
\cmidrule(lr){4-7} \cmidrule(lr){8-11} \cmidrule(lr){12-15}
 & & & \multicolumn{2}{c}{MB} & \multicolumn{2}{c}{MAE} & \multicolumn{2}{c}{MB} & \multicolumn{2}{c}{MAE} & \multicolumn{2}{c}{MB} & \multicolumn{2}{c}{MAE} \\
\cmidrule(lr){4-5} \cmidrule(lr){6-7} \cmidrule(lr){8-9} \cmidrule(lr){10-11} \cmidrule(lr){12-13} \cmidrule(lr){14-15}
 & & & Univ. & Joint & Univ. & Joint & Univ. & Joint & Univ. & Joint & Univ. & Joint & Univ. & Joint \\
\midrule
\multirow{9}{*}{1}  & 0.1 & 0.1 & 0.002 & -0.004 & 0.754 & 0.729 & -0.091 & 0.006 & 0.213 & 0.160 & 0.013 & -0.001 & 1.106 & 1.100 \\
 & 0.1 & 0.5 & -0.000 & -0.001 & 0.750 & 0.731 & -0.090 & 0.004 & 0.212 & 0.160 & 0.014 & 0.006 & 1.099 & 1.095 \\
 & 0.1 & 0.9 & -0.000 & 0.001 & 0.751 & 0.712 & -0.094 & -0.004 & 0.213 & 0.157 & 0.013 & 0.003 & 1.089 & 1.088 \\
 & 0.5 & 0.1 & -0.011 & -0.014 & 0.745 & 0.711 & -0.092 & 0.009 & 0.214 & 0.159 & 0.012 & 0.007 & 1.101 & 1.092 \\
 & 0.5 & 0.5 & -0.010 & -0.020 & 0.745 & 0.716 & -0.091 & 0.005 & 0.212 & 0.159 & 0.013 & 0.005 & 1.095 & 1.088 \\
 & 0.5 & 0.9 & -0.000 & -0.018 & 0.740 & 0.702 & -0.091 & -0.004 & 0.213 & 0.156 & 0.015 & -0.001 & 1.089 & 1.082 \\
 & 0.9 & 0.1 & 0.002 & -0.025 & 0.730 & 0.668 & -0.093 & 0.006 & 0.213 & 0.159 & 0.015 & 0.006 & 1.096 & 1.087 \\
 & 0.9 & 0.5 & -0.019 & -0.042 & 0.731 & 0.686 & -0.088 & 0.005 & 0.212 & 0.158 & 0.017 & 0.005 & 1.093 & 1.084 \\
 & 0.9 & 0.9 & -0.009 & -0.017 & 0.731 & 0.671 & -0.094 & -0.005 & 0.213 & 0.155 & 0.008 & 0.001 & 1.087 & 1.076 \\
\midrule
\multirow{9}{*}{0.5}  & 0.1 & 0.1 & -0.012 & -0.019 & 0.748 & 0.733 & -0.094 & 0.008 & 0.212 & 0.160 & 0.015 & 0.002 & 1.084 & 1.078 \\
 & 0.1 & 0.5 & 0.002 & -0.003 & 0.746 & 0.731 & -0.089 & 0.006 & 0.210 & 0.158 & 0.011 & 0.006 & 1.079 & 1.072 \\
 & 0.1 & 0.9 & -0.006 & -0.011 & 0.742 & 0.702 & -0.096 & -0.003 & 0.215 & 0.156 & 0.012 & 0.000 & 1.071 & 1.065 \\
 & 0.5 & 0.1 & -0.003 & -0.013 & 0.739 & 0.708 & -0.086 & 0.007 & 0.208 & 0.159 & 0.018 & 0.006 & 1.084 & 1.071 \\
 & 0.5 & 0.5 & -0.011 & -0.012 & 0.741 & 0.711 & -0.089 & 0.006 & 0.210 & 0.159 & 0.013 & 0.006 & 1.076 & 1.073 \\
 & 0.5 & 0.9 & -0.017 & -0.023 & 0.730 & 0.693 & -0.089 & -0.002 & 0.209 & 0.155 & 0.013 & -0.001 & 1.065 & 1.061 \\
 & 0.9 & 0.1 & -0.018 & -0.029 & 0.727 & 0.669 & -0.093 & 0.006 & 0.210 & 0.158 & 0.015 & 0.005 & 1.080 & 1.067 \\
 & 0.9 & 0.5 & -0.006 & -0.021 & 0.727 & 0.677 & -0.087 & 0.006 & 0.208 & 0.157 & 0.014 & 0.006 & 1.071 & 1.064 \\
 & 0.9 & 0.9 & -0.008 & -0.020 & 0.720 & 0.662 & -0.090 & -0.002 & 0.209 & 0.155 & 0.015 & 0.002 & 1.062 & 1.058 \\
\midrule
\multirow{9}{*}{0.25}  & 0.1 & 0.1 & -0.014 & -0.018 & 0.742 & 0.728 & -0.095 & 0.006 & 0.211 & 0.160 & 0.013 & 0.004 & 1.075 & 1.070 \\
& 0.1 & 0.5 & -0.009 & -0.005 & 0.750 & 0.723 & -0.087 & 0.008 & 0.208 & 0.159 & 0.012 & 0.006 & 1.070 & 1.064 \\
& 0.1 & 0.9 & -0.010 & -0.014 & 0.735 & 0.698 & -0.092 & -0.002 & 0.209 & 0.156 & 0.013 & -0.001 & 1.059 & 1.055 \\
& 0.5 & 0.1 & -0.003 & -0.013 & 0.738 & 0.701 & -0.091 & 0.009 & 0.209 & 0.158 & 0.015 & 0.007 & 1.072 & 1.060 \\
& 0.5 & 0.5 & -0.006 & -0.015 & 0.737 & 0.707 & -0.091 & 0.005 & 0.208 & 0.158 & 0.012 & 0.008 & 1.065 & 1.056 \\
& 0.5 & 0.9 & -0.010 & -0.017 & 0.732 & 0.693 & -0.089 & -0.002 & 0.208 & 0.155 & 0.013 & -0.000 & 1.056 & 1.050 \\
& 0.9 & 0.1 & 0.003 & -0.022 & 0.722 & 0.658 & -0.086 & 0.007 & 0.206 & 0.158 & 0.015 & 0.004 & 1.068 & 1.057 \\
& 0.9 & 0.5 & -0.014 & -0.026 & 0.724 & 0.669 & -0.091 & 0.006 & 0.207 & 0.157 & 0.012 & 0.008 & 1.065 & 1.054 \\
& 0.9 & 0.9 & -0.028 & -0.028 & 0.720 & 0.660 & -0.088 & -0.002 & 0.206 & 0.154 & 0.009 & 0.003 & 1.052 & 1.044 \\
\midrule
\multirow{9}{*}{0.1} & 0.1 & 0.1 & -0.012 & -0.018 & 0.744 & 0.728 & -0.089 & 0.007 & 0.206 & 0.160 & 0.011 & 0.002 & 1.066 & 1.061 \\
& 0.1 & 0.5 & -0.010 & -0.015 & 0.744 & 0.723 & -0.089 & 0.006 & 0.206 & 0.158 & 0.013 & 0.006 & 1.061 & 1.054 \\
& 0.1 & 0.9 & -0.011 & -0.010 & 0.739 & 0.693 & -0.092 & -0.001 & 0.208 & 0.155 & 0.012 & 0.002 & 1.055 & 1.048 \\
& 0.5 & 0.1 & -0.010 & -0.027 & 0.736 & 0.707 & -0.088 & 0.007 & 0.206 & 0.159 & 0.014 & 0.004 & 1.064 & 1.054 \\
& 0.5 & 0.5 & -0.005 & -0.009 & 0.736 & 0.706 & -0.089 & 0.005 & 0.206 & 0.158 & 0.015 & 0.006 & 1.057 & 1.050 \\
& 0.5 & 0.9 & -0.013 & -0.014 & 0.732 & 0.684 & -0.088 & -0.001 & 0.205 & 0.155 & 0.011 & 0.001 & 1.048 & 1.040 \\
& 0.9 & 0.1 & -0.008 & -0.025 & 0.723 & 0.660 & -0.088 & 0.005 & 0.205 & 0.158 & 0.011 & 0.004 & 1.059 & 1.049 \\
& 0.9 & 0.5 & -0.022 & -0.031 & 0.722 & 0.672 & -0.088 & 0.006 & 0.206 & 0.158 & 0.014 & 0.009 & 1.058 & 1.051 \\
& 0.9 & 0.9 & -0.015 & -0.027 & 0.724 & 0.657 & -0.091 & -0.000 & 0.207 & 0.154 & 0.012 & 0.004 & 1.048 & 1.037 \\
\bottomrule
\end{tabular}
}
\end{table}

Table~\ref{tab:bias_mae_redhot_n300_wide} reports the median bias and median
absolute error for the separate univariate models (Univ.) and the joint
bivariate model (Joint) at $n_L = 300$; the corresponding
results for $n_L = 600$ and $n_L = 1200$ are given
in Appendix~\ref{app:sim_results}. We find that the pattern of results differs across the three quantities considered, indicating where the benefit of joint modelling accrues most.

For the spatial field $S$, the joint model attains a lower median absolute
error in every scenario, and the size of this advantage increases with the cross-field correlation $\rho_S$.
At $n_L = 300$ and $c = 1$, the MAE falls from
$0.754$ to $0.729$ when $\rho_S = 0.1$, a reduction of around 3\%, but from
$0.730$ to $0.668$ when $\rho_S = 0.9$, a reduction of around 9\%. This
relative advantage is stable across sample sizes, while the absolute error
falls for both models as $n_{L}$ increases. The joint model's ability to borrow
spatial information across antigens therefore yields a genuine efficiency
gain whose magnitude is governed by how much exploitable cross-field signal
is present, rather than by the amount of data available.

For the latent seroreactivity process $T$, the univariate models exhibit a
negative median bias of approximately $-0.09$ that is essentially constant
across every combination of $\rho_S$, $\rho_T$, $c$ and $n_{L}$, while the joint
model's bias remains close to zero throughout.  The MAE for $T$ is approximately 20--25\% lower under the joint model at
every scenario and every sample size considered. This consistent finding
suggests that the principal benefit of jointly modelling the two antigens
lies in more accurate recovery of each individual's underlying immune activation state.

For the observed antibody outcome $Y$, both models are close to unbiased and
the joint model retains a consistent but small advantage in MAE. Reducing the
measurement-error variance through $c$ lowers the absolute error for both
models, as expected, but leaves the difference between them largely unchanged:
at $n_{L} = 300$ and $\rho_S = \rho_T = 0.9$, the gap in MAE is $0.011$ at $c = 1$
and $0.011$ at $c = 0.1$. The improvement in the recovery of $T$ therefore
propagates only weakly to $Y$, and this attenuation is not 
attributable to the magnitude of the measurement error alone.

\section{Discussion}
\label{sec:discussion}
We have introduced a modelling framework for the joint analysis of
antibody responses to multiple antigens, based on three guiding
principles that determine how different sources of dependence should be accounted for in the model. More specifically, the proposed modelling framework accounts for
three sources of dependence: dependence induced by a shared
infection event, which stimulates responses to several antigens
simultaneously within an individual; dependence induced by a
shared exposure environment, which correlates infection risk
between individuals sampled at nearby locations and, through this
shared risk, across the antigens to which they have been exposed;
and dependence arising from individual-specific immunological
characteristics that shape how a person's immune system responds
once exposed. To achieve this, we have focused primarily on the
bivariate case, extending the univariate latent variable framework
of \citet{giorgiwallin2026} to a joint distribution for a pair of
antigen-specific continuous seroreactivity processes. Spatial
dependence between the two processes is introduced through a
bivariate Mat\'ern random field acting on the mixing structure of
this joint distribution. The guiding principles set out in
Section~\ref{sec:framework} also provide a general basis for
extending the framework to settings with more than two antibody
responses, while guarding against excessive model complexity and
keeping the resulting structure interpretable.

The results of Sections~\ref{sec:tv}--\ref{sec:simulation} indicate that the
inferential benefits accrued by the joint model depend on which feature of
the antibody distribution is of interest. In the application, the joint model
gave a markedly more reliable and less uncertain approximation of the
bivariate antibody distribution than the two univariate models, whereas the
two performed comparably on recovering each antibody's marginal distribution.
Cross-validated predictive comparisons showed larger differences in the scoring rules
 between spatial and non-spatial models than between joint and separate univariate specifications. 
 The benefit of joint rather than separate univariate modelling was instead most evident in the bivariate 
 highest-density predictive regions: at essentially the same coverage, the
 mean area was $12.60\%$ smaller under the joint spatial model. The marked
 reduction in HDR area contrasts with the very small difference in energy
 score between the joint and separate models. Although the energy score is
 proper for the joint distribution, this contrast indicates limited
 sensitivity to dependence differences in this application and motivates
 complementing it with scoring rules that target the dependence structure
 more directly.
Joint modelling is
therefore of clear benefit when the target of inference is a joint property
of the antibody distributions, such as cross-antigen dependence itself or the
probability that an individual is highly seroreactive to both antigens
simultaneously. When interest lies solely in the marginal distributions of the antibody responses, fitting a joint model brings little gain, and its added complexity may even inflate the uncertainty of the marginal parameter estimates. The simulation study allowed to  locate the
source of this benefit more precisely. The largest and most consistent gains
are in the recovery of the latent seroreactivity process $T$, where the
univariate models exhibit a more substantial bias than a joint model. Improvements in the recovery of the spatial fields are more
modest and scale with the strength of the cross-field correlation, and they
propagate only weakly to the observed antibody concentrations, which are
separated from $T$ by the measurement error of the observation model. The
case for joint modelling is thus strongest when the latent seroreactivity
process, rather than the observed concentration, is the quantity of
epidemiological interest.

Although we have illustrated the framework for the bivariate case only, the
underlying construction extends naturally to more than two antigens. For
example, it could be used to address the joint modelling of multiplex
antibody responses considered by \citet{baudemont2026}, who extend
serocatalytic models to combinations of up to three \textit{P.\ falciparum}
antigens using a serological incidence shared across antigens, together with
antigen-specific probabilities of seroconversion given exposure and
antigen-specific seroreversion rates. Cross-antigen dependence is modelled
through the shared exposure process, since a single exposure event may
seroconvert both antigens at once, with probability given by the product of
their seroconversion probabilities. In terms of the principles of
Section~\ref{sec:framework}, this construction respects \textbf{P1}, in that
all dependence is placed in the serostatus process, but does not respect
\textbf{P2}, since dichotomisation removes response magnitude from the model
altogether, and addresses only the shared-infection-event mechanism of
\textbf{P3}.

Within our framework, the specification of \citet{baudemont2026} could be
extended by using their shared incidence parameter to define a mechanistic
form for the mixing probabilities $\pi_i(\mv{z})$, following
\citet{giorgiwallin2026}. Let $\lambda(\mv{x}_i)$
denote a spatially-varying exposure intensity, with
\begin{equation*}
\log\{\lambda(\mv{x}_i)\} := \mv{d}_i^\top\beta + S(\mv{x}_i),
\end{equation*}
where $\beta$ collects covariate effects analogous to those already entering
$\mv{d}_i$ in Sections~\ref{sec:model} and~\ref{sec:spatial}, and
$S(\mv{x}_i)$ is a spatial random field of the kind introduced in
Section~\ref{sec:spatial}; let $\gamma_k \in (0,1)$ denote the probability
that an exposure event moves an individual into the high-seroreactivity
regime for antigen $k$; and let $\varrho_k > 0$ denote the rate of reversion
to the low-seroreactivity regime. The marginal mixing probability for antigen
$k$ then becomes
\begin{equation*}
p_{i,k} := \frac{\lambda(\mv{x}_i)\gamma_k}{\lambda(\mv{x}_i)\gamma_k+\varrho_k}
\left[1-\exp\{-(\lambda(\mv{x}_i)\gamma_k+\varrho_k)a_i\}\right],
\end{equation*}
which has the same reversible-catalytic form as the age-specific solution of
\citet{baudemont2026} under constant transmission, now made spatially
explicit through $\lambda(\mv{x}_i)$. Conditional probabilities for the
second and third antigens are obtained by adding age-dependent association
terms $\delta_{12}(a_i)$, $\delta_{13}(a_i)$ and $\delta_{23}(a_i)$ to the
logit of this expression, exactly as $\delta(a_i)$ is added
in~\eqref{eq:mixing_probs}. All three principles are then respected, since the mechanistic specification
enters only the mixing probabilities, leaving the observation model unchanged
(\textbf{P1}) and confining exposure to $\lambda(\mv{x}_i)$ and $\gamma_k$
while response magnitude remains governed by the means of the mixture's components
(\textbf{P2}). The three mechanisms of \textbf{P3} are represented, respectively, by the
$\delta_{jk}(a_i)$, by the shared $S(\mv{x}_i)$, and by $\rho_T$. This reformulation rests, however, on a single
shared exposure intensity and a single spatial field $S(\mv{x})$ common to
all antigens, rather than the antigen-specific fields $S_1(\mv{x})$ and
$S_2(\mv{x})$ of Section~\ref{sec:spatial}. In contexts where this assumption
is biologically plausible, it could provide the basis for a more parsimonious
and interpretable class of models.

The extensions illustrated above are relatively straightforward when
moving from two to several more, say up to five, antigens from the
same pathogen, as in the mechanistic reformulation discussed above.
More general extensions, however, whether to a larger number of
antigens or to antigens spanning multiple distinct pathogens, raise
epidemiological and statistical challenges that we have not addressed
here. The
bivariate model developed in this paper already illustrates the
substantial complexity involved in jointly representing latent
seroreactivity and its spatial structure, and this complexity is only
likely to increase further as more antigens are added. A general
multivariate formulation would need to be tailored to the specific
epidemiological context, and would require careful consideration of
how best to incorporate biological knowledge into the model
formulation. In particular, whether the antigens under consideration
belong to a single pathogen or to multiple distinct pathogens has
direct implications for how cross-antigen association should be
modelled. To further extend the proposed bivariate framework to a
multivariate context, two distinct approaches should be considered.
The first is a disease-agnostic approach in which no prior structure
is imposed on the correlation structure to quantify the association
among antibody responses, allowing the data to determine the full
pattern of dependence. Alternatively, an assumption-driven approach
would instead impose a parsimonious correlation structure informed by biological knowledge, for example by allowing association only among antigens sharing a common pathogen or transmission route. In
addition, the computational approximations used here, while adequate
for a bivariate outcome, are unlikely to scale up to settings with
tens of antibody responses. Future work should therefore pursue
additional case studies investigating both disease-agnostic and
biologically motivated parsimonious correlation structures in a
genuinely multivariate setting, together with the development of
computational approaches better suited to scaling beyond the
bivariate case considered here.

The three principles set out in Section~\ref{sec:framework} serve to
regulate model complexity while anchoring the resulting model structure in biological reasoning. However, these modelling principles could be modified depending on the context. \textbf{P1} may be not be valid in cases where antibody cross-reactivity
between the two antigens is biologically plausible, since this induces correlation between the observed antibody concentrations that is unrelated to the shared exposure and immunological mechanisms represented
by our latent formulation. In this case, such
correlation should ideally be represented as a residual covariance term
directly in the observation model~\eqref{eq:obs_joint}, rather than
absorbed into the joint distribution of $\mv{T}_i$. In practice, however,
separating assay-induced cross-reactivity from genuine biological
correlation in seroreactivity is unlikely to be feasible without further
assumptions or external information, such as replicate measurements or known cross-reactivity profiles for the antigens in question. \textbf{P2}, while also an assumption, reflects a pragmatic
choice; a fully general model in which the spatial field enters both the
component locations in~\eqref{eq:biv_latentT_location_links} and the
mixing probabilities would face a further identifiability problem. The restriction imposed by \textbf{P2} is
additionally supported by the fact that this formulation has performed
satisfactorily in the application presented here and in other (unpublished) analyses, suggesting that confining spatial variation to
exposure probability captures the dominant spatial signal in seroreactivity
data of this kind. \textbf{P3} is stated for antigens sharing a pathogen
and life-cycle stage; for antigens from distinct pathogens the
shared-infection-event pathway no longer applies and the constrained form
of $\delta(a_i)$ could be replaced, for example, by an unconstrained specification. 

Finally, the use of the Laplace approximation makes the proposed model computationally feasible,
but the parametric bootstrap used for uncertainty quantification remains
costly. A substantially cheaper alternative is to estimate parameter
uncertainty from the inverse Hessian of the Laplace-approximated marginal
log-likelihood, as implemented in TMB \citep{kristensen2016tmb}. The validity
of such inference, however, depends on the accuracy of the likelihood
approximation. \citet{ogden2017naive} gives conditions, in terms of errors in
the approximate score and observed information, under which approximate-
likelihood inference retains the first-order properties of exact likelihood
inference. This issue is particularly relevant here because the dimension of
the FEM latent field increases with the discretisation
\citep{ogden2021error}. Moreover, \citet{han2024enhanced} demonstrate that
ordinary Laplace-based standard errors can underestimate sampling variability
in spatial random-effects models. Establishing the corresponding validity
conditions for our model, or using enhanced Laplace variance estimation,
could therefore provide a substantially cheaper alternative to the
parametric bootstrap.

\bibliographystyle{plainnat}
\bibliography{references}

@article{giorgiwallin2026,
  author =        {Giorgi, E. and Wallin, J.},
  journal =       {Biostatistics},
  month =         {01},
  number =        {1},
  pages =         {kxag008},
  title =         {A flexible class of latent variable models for the
                   analysis of antibody response data},
  volume =        {27},
  year =          {2026},
  doi =           {10.1093/biostatistics/kxag008},
  issn =          {1468-4357},
  url =           {https://doi.org/10.1093/biostatistics/kxag008},
}

@article{metcalf2016,
  author =        {Metcalf, C. J. E. and Farrar, J. and Cutts, F. T. and
                   Basta, N. E. and Graham, A. L. and Lessler, J. and
                   Ferguson, N. M. and Burke, D. S. and Grenfell, B. T.},
  journal =       {The Lancet},
  number =        {10045},
  pages =         {728--730},
  title =         {Use of serological surveys to generate key insights
                   into the changing global landscape of infectious
                   disease},
  volume =        {388},
  year =          {2016},
  doi =           {10.1016/S0140-6736(16)30164-7},
}

@article{corran2007,
  author =        {Corran, P. H. and Coleman, P. G. and Riley, E. M. and
                   Drakeley, C. J.},
  journal =       {Trends in Parasitology},
  number =        {12},
  pages =         {575--582},
  title =         {Serology: a robust indicator of malaria transmission
                   intensity?},
  volume =        {23},
  year =          {2007},
  doi =           {10.1016/j.pt.2007.08.023},
}

@article{Drakeley2005,
  author =        {Drakeley, C. J. and Corran, P. H. and Coleman, P. G. and
                   Tongren, J. E. and McDonald, S. L. R. and
                   Carneiro, I. and Malima, R. and Lusingu, J. and
                   Manjurano, A. and Nkya, W. M. M. and Lemnge, M. M. and
                   Cox, J. and Reyburn, H. and Riley, E. M.},
  journal =       {Proceedings of the National Academy of Sciences},
  number =        {14},
  pages =         {5108--5113},
  title =         {Estimating medium- and long-term trends in malaria
                   transmission by using serological markers of malaria
                   exposure},
  volume =        {102},
  year =          {2005},
  doi =           {10.1073/pnas.0408725102},
}

@article{arnold2018,
  author =        {Arnold, B. F. and Scobie, H. M. and Priest, J. W. and
                   Lammie, P. J.},
  journal =       {Emerging Infectious Diseases},
  number =        {7},
  pages =         {1188--1194},
  title =         {Integrated serologic surveillance of population
                   immunity and disease transmission},
  volume =        {24},
  year =          {2018},
  doi =           {10.3201/eid2407.171928},
}

@article{Carcelen2025,
  author =        {Carcelen, A. C. and Monjane, C. and
                   B{\'e}rub{\'e}, S. and Takahashi, S. and
                   Sultane, T. and Chelene, I. and Cooley, G. and
                   Goodhew, E. B. and Patterson, C. and Tetteh, K. and
                   Mutambe, M. and Higdon, M. M. and Mwinnyaa, G. and
                   Nhapure, G. and Duce, P. and Martin, D. L. and
                   Drakeley, C. and Moss, W. J. and Macicame, I.},
  journal =       {Nature Communications},
  pages =         {7946},
  title =         {Multiplex bead assays enable integrated serological
                   surveillance and reveal cross-pathogen
                   vulnerabilities in {Zambezia} Province, {Mozambique}},
  volume =        {16},
  year =          {2025},
  doi =           {10.1038/s41467-025-62305-9},
}

@article{royston2006,
  author =        {Royston, P. and Altman, D. G. and Sauerbrei, W.},
  journal =       {Statistics in Medicine},
  number =        {1},
  pages =         {127--141},
  title =         {Dichotomizing continuous predictors in multiple
                   regression: a bad idea},
  volume =        {25},
  year =          {2006},
  doi =           {10.1002/sim.2331},
}

@article{DiggleGiorgi2016,
  author =        {Diggle, P. J. and Giorgi, E.},
  journal =       {Journal of the American Statistical Association},
  number =        {515},
  pages =         {1096--1120},
  title =         {Model-based geostatistics for prevalence mapping in
                   low-resource settings},
  volume =        {111},
  year =          {2016},
  doi =           {10.1080/01621459.2015.1123158},
}

@article{Stresman2017,
  author =        {Stresman, G. H. and Giorgi, E. and Baidjoe, A. and
                   Knight, P. and Odongo, W. and Owaga, C. and
                   Shagari, S. and Makori, E. and Stevenson, J. and
                   Drakeley, C. and Cox, J. and Bousema, T. and
                   Diggle, P. J.},
  journal =       {Scientific Reports},
  pages =         {45849},
  title =         {Impact of metric and sample size on determining
                   malaria hotspot boundaries},
  volume =        {7},
  year =          {2017},
  doi =           {10.1038/srep45849},
}

@article{Sasanami2023,
  author =        {Sasanami, M. and Amoah, B. and Diori, A. N. and
                   Amza, A. and Souley, A. S. Y. and Bakhtiari, A. and
                   Kadri, B. and Szwarcwald, C. L. and
                   Ferreira Gomez, D. V. and Almou, I. and
                   Lopes, M. F. C. and Masika, M. P. and Beidou, N. and
                   Boyd, S. and Harding-Esch, E. M. and Solomon, A. W. and
                   Giorgi, E.},
  journal =       {PLOS Neglected Tropical Diseases},
  number =        {7},
  pages =         {e0011476},
  title =         {Using model-based geostatistics for assessing the
                   elimination of trachoma},
  volume =        {17},
  year =          {2023},
  doi =           {10.1371/journal.pntd.0011476},
}

@article{efron1979bootstrap,
  title={Bootstrap methods: another look at the jackknife},
  author={Efron, B.},
  journal={The Annals of Statistics},
  volume={7},
  number={1},
  pages={1--26},
  year={1979},
  publisher={Institute of Mathematical Statistics}
}

@article{plackett1965,
  author  = {Plackett, R. L.},
  title   = {A Class of Bivariate Distributions},
  journal = {Journal of the American Statistical Association},
  volume  = {60},
  number  = {310},
  pages   = {516--522},
  year    = {1965},
  doi     = {10.1080/01621459.1965.10480807}
}

@article{CadavidRestrepo2023,
  author =        {Cadavid Restrepo, A. M. and Martin, B. M. and
                   Fuimaono, S. and Clements, A. C. A. and Graves, P. M. and
                   Lau, C. L.},
  journal =       {PLOS Neglected Tropical Diseases},
  number =        {7},
  pages =         {e0010840},
  title =         {Spatial predictive risk mapping of lymphatic
                   filariasis residual hotspots in {A}merican {S}amoa
                   using demographic and environmental factors},
  volume =        {17},
  year =          {2023},
  doi =           {10.1371/journal.pntd.0010840},
}

@article{lubyayi2021,
  author  = {Lubyayi, L. and Mawa, P. A. and Cose, S. and Elliott, A. M. and Levin, J. and Webb, E. L.},
  title   = {Analysis of multivariate longitudinal immuno-epidemiological data using a pairwise joint modelling approach},
  journal = {BMC Immunology},
  volume  = {22},
  pages   = {63},
  year    = {2021},
  doi     = {10.1186/s12865-021-00453-5}
}

@article{hay2024serodynamics,
  author  = {Hay, J. A. and Routledge, I. and Takahashi, S.},
  title   = {Serodynamics: A primer and synthetic review of methods for epidemiological inference using serological data},
  journal = {Epidemics},
  volume  = {49},
  pages   = {100806},
  year    = {2024},
  doi     = {10.1016/j.epidem.2024.100806}
}

@article{ODriscoll2025,
  author =        {O'Driscoll, M. and Hoz{\'e}, N. and Lefrancq, N. and
                   Ribeiro dos Santos, G. and Hoinard, D. and
                   Rahman, M. Z. and Paul, K. K. and Titu, A. M. N. and
                   Alam, M. S. and Hossain, M. E. and Vanhomwegen, J. and
                   Cauchemez, S. and Gurley, E. S. and Salje, H.},
  journal =       {Science Translational Medicine},
  number =        {826},
  pages =         {eads8680},
  title =         {Epidemiological and antigenic inferences from
                   serological cross-reactivity among arboviruses},
  volume =        {17},
  year =          {2025},
  doi =           {10.1126/scitranslmed.ads8680},
}

@article{kleiber2017coherence,
  author =        {Kleiber, W.},
  journal =       {Statistica Sinica},
  number =        {4},
  pages =         {1675--1697},
  title =         {Coherence for multivariate random fields},
  volume =        {27},
  year =          {2017},
  doi =           {10.5705/ss.202015.0309},
}

@article{bolin2020multivariate,
  author =        {Bolin, D. and Wallin, J.},
  journal =       {Journal of the Royal Statistical Society Series B:
                   Statistical Methodology},
  number =        {1},
  pages =         {215--239},
  title =         {Multivariate type G Mat{\'e}rn stochastic partial
                   differential equation random fields},
  volume =        {82},
  year =          {2020},
  doi =           {10.1111/rssb.12351},
}

@article{whittle1963stochastic,
  author =        {Whittle, P.},
  journal =       {Bulletin of the International Statistical Institute},
  number =        {2},
  pages =         {974--994},
  title =         {Stochastic-processes in several dimensions},
  volume =        {40},
  year =          {1963},
}

@article{gneiting2010matern,
  author =        {Gneiting, T. and Kleiber, W. and Schlather, M.},
  journal =       {Journal of the American Statistical Association},
  number =        {491},
  pages =         {1167--1177},
  title =         {Mat{\'e}rn cross-covariance functions for
                   multivariate random fields},
  volume =        {105},
  year =          {2010},
  doi =           {10.1198/jasa.2010.tm09420},
}

@article{lindgren2011explicit,
  author =        {Lindgren, F. and Rue, H. and Lindstr{\"o}m, J.},
  journal =       {Journal of the Royal Statistical Society Series B:
                   Statistical Methodology},
  number =        {4},
  pages =         {423--498},
  title =         {An explicit link between Gaussian fields and Gaussian
                   Markov random fields: the stochastic partial
                   differential equation approach},
  volume =        {73},
  year =          {2011},
  doi =           {10.1111/j.1467-9868.2011.00777.x},
}

@manual{lindgren2024fmesher,
  author =        {Lindgren, F.},
  note =          {R package version 0.8.0},
  title =         {{fmesher}: Triangle Meshes and Related Geometry
                   Tools},
  year =          {2026},
  doi =           {10.32614/CRAN.package.fmesher},
  url =           {https://CRAN.R-project.org/package=fmesher},
}

@manual{bolin2024rspde,
  author =        {Bolin, D. and Simas, A. B.},
  note =          {R package version 2.5.2},
  title =         {{rSPDE}: Rational Approximations of Fractional
                   Stochastic Partial Differential Equations},
  year =          {2026},
  doi =           {10.32614/CRAN.package.rSPDE},
  url =           {https://CRAN.R-project.org/package=rSPDE},
}

@article{tierney1986accurate,
  author =        {Tierney, L. and Kadane, J. B.},
  journal =       {Journal of the American Statistical Association},
  number =        {393},
  pages =         {82--86},
  title =         {Accurate Approximations for Posterior Moments and
                   Marginal Densities},
  volume =        {81},
  year =          {1986},
  doi =           {10.1080/01621459.1986.10478240},
}

@book{rueheld2005,
  address =       {Boca Raton, FL},
  author =        {Rue, H. and Held, L.},
  publisher =     {Chapman \& Hall/CRC},
  series =        {Monographs on Statistics and Applied Probability},
  title =         {Gaussian Markov Random Fields: Theory and
                   Applications},
  volume =        {104},
  year =          {2005},
  doi =           {10.1201/9780203492024},
  isbn =          {978-1-58488-432-3},
}

@article{chen2008cholmod,
  author =        {Chen, Y. and Davis, T. A. and Hager, W. W. and
                   Rajamanickam, S.},
  journal =       {ACM Transactions on Mathematical Software},
  number =        {3},
  pages =         {1--14},
  title =         {Algorithm 887: {CHOLMOD}, Supernodal Sparse
                   {Cholesky} Factorization and Update/Downdate},
  volume =        {35},
  year =          {2008},
  doi =           {10.1145/1391989.1391995},
}

@article{kristensen2016tmb,
  author =        {Kristensen, K. and Nielsen, A. and Berg, C. W. and
                   Skaug, H. and Bell, B. M.},
  journal =       {Journal of Statistical Software},
  number =        {5},
  pages =         {1--21},
  title =         {{TMB}: Automatic Differentiation and {Laplace}
                   Approximation},
  volume =        {70},
  year =          {2016},
  doi =           {10.18637/jss.v070.i05},
}

@article{Tatem2017,
  author =        {Tatem, A. J.},
  journal =       {Scientific Data},
  pages =         {170004},
  title =         {WorldPop, open data for spatial demography},
  volume =        {4},
  year =          {2017},
  doi =           {10.1038/sdata.2017.4},
}

@article{Bousema2013,
  author =        {Bousema, T. and Griffin, J. T. and Sauerwein, R. W. and
                   Smith, D. L. and Churcher, T. S. and Takken, W. and
                   Ghani, A. C. and Drakeley, C. and Gosling, R.},
  journal =       {Trials},
  pages =         {36},
  title =         {The impact of hotspot-targeted interventions on
                   malaria transmission: study protocol for a
                   cluster-randomized controlled trial},
  volume =        {14},
  year =          {2013},
  doi =           {10.1186/1745-6215-14-36},
  url =           {https://doi.org/10.1186/1745-6215-14-36},
}

@article{Bousema2016,
  author =        {Bousema, T. and Stresman, G. and Baidjoe, A. Y. and
                   Stevenson, J. and Omedo, I. and Osoti, V. and
                   Macharia, A. and Ouma, P. and Yaa, P. and Jacobs, E. and
                   Cook, J. and Kleinschmidt, I. and Thomas, M. and
                   Drakeley, C. and Cox, J. and Alaii, J. and Odongo, W. and
                   Laserson, K. and Kariuki, S. and Slutsker, L. and
                   Desai, M. and Barger, B. and Tiono, A. B. and
                   Sauerwein, R. W. and Ogutu, B. and Gosling, R.},
  journal =       {PLoS Medicine},
  number =        {4},
  pages =         {e1001993},
  title =         {The impact of hotspot-targeted interventions on
                   malaria transmission in Rachuonyo South District in
                   the Western Kenyan Highlands: a cluster-randomized
                   controlled trial},
  volume =        {13},
  year =          {2016},
  doi =           {10.1371/journal.pmed.1001993},
  url =           {https://doi.org/10.1371/journal.pmed.1001993},
}

@article{Beeson2016,
  author =        {Beeson, J. G. and Drew, D. R. and Boyle, M. J. and
                   Feng, G. and Fowkes, F. J. I. and Richards, J. S.},
  journal =       {FEMS Microbiology Reviews},
  number =        {3},
  pages =         {343--372},
  title =         {Merozoite surface proteins in red blood cell
                   invasion, immunity and vaccines against malaria},
  volume =        {40},
  year =          {2016},
  doi =           {10.1093/femsre/fuw001},
}

@article{Cowman2017,
  author =        {Cowman, A. F. and Tonkin, C. J. and Tham, W.-H. and
                   Duraisingh, M. T.},
  journal =       {Cell Host \& Microbe},
  number =        {2},
  pages =         {232--245},
  title =         {The molecular basis of erythrocyte invasion by
                   malaria parasites},
  volume =        {22},
  year =          {2017},
  doi =           {10.1016/j.chom.2017.07.003},
}

@article{scheffe1947,
  author =        {Scheff\'e, H.},
  journal =       {The Annals of Mathematical Statistics},
  number =        {3},
  pages =         {434--438},
  title =         {A Useful Convergence Theorem for Probability
                   Distributions},
  volume =        {18},
  year =          {1947},
  doi =           {10.1214/aoms/1177730390},
}

@article{gibbs2002,
  author =        {Gibbs, A. L. and Su, F. E.},
  journal =       {International Statistical Review},
  number =        {3},
  pages =         {419--435},
  title =         {On Choosing and Bounding Probability Metrics},
  volume =        {70},
  year =          {2002},
  doi =           {10.1111/j.1751-5823.2002.tb00178.x},
}

@article{csiszar1967,
  author =        {Csisz\'ar, I.},
  journal =       {Studia Scientiarum Mathematicarum Hungarica},
  pages =         {299--318},
  title =         {Information-Type Measures of Difference of
                   Probability Distributions and Indirect Observations},
  volume =        {2},
  year =          {1967},
}

@article{gneiting2007strictly,
  author =        {Gneiting, T. and Raftery, A. E.},
  journal =       {Journal of the American Statistical Association},
  number =        {477},
  pages =         {359--378},
  title =         {Strictly Proper Scoring Rules, Prediction, and
                   Estimation},
  volume =        {102},
  year =          {2007},
  doi =           {10.1198/016214506000001437},
}

@article{bolinwallin2023,
  author =        {Bolin, D. and Wallin, J.},
  journal =       {Statistical Science},
  number =        {1},
  pages =         {140--159},
  title =         {Local Scale Invariance and Robustness of Proper
                   Scoring Rules},
  volume =        {38},
  year =          {2023},
  doi =           {10.1214/22-STS864},
}

@article{hyndman1996,
  author =        {Hyndman, R. J.},
  journal =       {The American Statistician},
  number =        {2},
  pages =         {120--126},
  title =         {Computing and Graphing Highest Density Regions},
  volume =        {50},
  year =          {1996},
  doi =           {10.1080/00031305.1996.10474359},
}

@book{silverman1986,
  address =       {London},
  author =        {Silverman, B. W.},
  publisher =     {Chapman and Hall},
  title =         {Density Estimation for Statistics and Data Analysis},
  year =          {1986},
}

@article{baudemont2026,
  author =        {Baudemont, G. and Obadia, T. and Garcia, L. and
                   Lambert, C. and Donnadieu, F. and Diene Sarr, F. and
                   Faye, J. and Sokhna, C. and Vigan-Womas, I. and
                   Toure-Balde, A. and Drakeley, C. and Niang, M. and
                   White, M. T.},
  journal =       {PLOS Computational Biology},
  number =        {7},
  pages =         {e1013630},
  title =         {Reconstruction of historical malaria transmission in
                   {S}enegal using multiplex serocatalytic models},
  volume =        {22},
  year =          {2026},
  doi =           {10.1371/journal.pcbi.1013630},
}

@article{ogden2017naive,
  author =        {Ogden, H. E.},
  journal =       {Biometrika},
  number =        {1},
  pages =         {153--164},
  title =         {On Asymptotic Validity of Naive Inference with an
                   Approximate Likelihood},
  volume =        {104},
  year =          {2017},
  doi =           {10.1093/biomet/asx002},
}

@article{ogden2021error,
  author =        {Ogden, Helen},
  journal =       {Stat},
  number =        {1},
  pages =         {e380},
  title =         {On the error in Laplace approximations of
                   high-dimensional integrals},
  volume =        {10},
  year =          {2021},
  doi =           {10.1002/sta4.380},
}

@article{han2024enhanced,
  author =        {Han, J. and Lee, Y.},
  journal =       {Journal of Multivariate Analysis},
  pages =         {105321},
  title =         {Enhanced {Laplace} Approximation},
  volume =        {202},
  year =          {2024},
  doi =           {10.1016/j.jmva.2024.105321},
}

@book{olver2010nist,
  address =       {Cambridge},
  editor =        {Olver, F. W. J. and Lozier, D. W. and Boisvert, R. F. and
                   Clark, C. W.},
  publisher =     {Cambridge University Press},
  title =         {NIST Handbook of Mathematical Functions},
  year =          {2010},
}

\appendix

\section*{Acknowledgements}

We thank all those who contributed to the collection of data included in this paper, specifically the survey participants in Kenya, and the KEMRI/CDC research team.

The computations described in this paper were performed using the University of Birmingham's BlueBEAR HPC service, which provides a High Performance Computing service to the University's research community. See \url{www.birmingham.ac.uk/bear} for more details.

\section{Proofs}
\label{app:proofs}
\begin{proof}[proof of Proposition~\ref{prop:biv-matern-validity}]
We begin by showing that the solution of the coupled system \eqref{eq:bispde}
has the multivariate covariance function $\mv{C}$ as
defined in the proposition.
Set $q_k := (\nu_{k}+1)/2$ and
$\bar\nu := q_1+q_2-1=(\nu_1+\nu_2)/2$. In two spatial
dimensions, the componentwise Whittle operators in
\eqref{eq:univ-spde-whittle} are
$\mathcal{L}_k := (\kappa_k^2-\Delta)^{q_k}$. The
dependence matrix in \eqref{eq:bispde} has inverse
\begin{equation*}
\mv{R} := \mv{D}(\rho_S)^{-1}
=
\begin{pmatrix}
1 & 0\\
\rho_S& \sqrt{1-\rho_S^2}
\end{pmatrix},
\qquad
\mv{R}\mv{R}^{\top}
=
\begin{pmatrix}
1 & \rho_S\\
\rho_S& 1
\end{pmatrix}.
\end{equation*}
The last matrix is positive definite because
$|\rho_S|<1$.

Following \citet[Proposition~1]{bolin2020multivariate}, Fourier transformation maps
$\tau_k\mathcal{L}_k$ to multiplication by
$\tau_k(\kappa_k^2+\|\boldsymbol\omega\|^2)^{q_k}$. Define the
diagonal multiplier matrix
\begin{equation*}
\mv{H}(\boldsymbol\omega)
:=
\operatorname{diag}\!\left\{
\tau_1(\kappa_1^2+\|\boldsymbol\omega\|^2)^{q_1},
\tau_2(\kappa_2^2+\|\boldsymbol\omega\|^2)^{q_2}
\right\}.
\end{equation*}
The two independent white-noise processes have constant spectral
density $(2\pi)^{-2}\mv{I}_2$. Consequently, the spectral density
matrix of the stationary solution is
\begin{equation*}
\mv{f}(\boldsymbol\omega)
=
\frac{1}{(2\pi)^2}
\mv{H}(\boldsymbol\omega)^{-1}
\mv{R}\mv{R}^{\top}
\mv{H}(\boldsymbol\omega)^{-1}.
\end{equation*}
Both $\mv{H}(\boldsymbol\omega)$ and $\mv{R}$ are invertible.
Therefore, for every non-zero $\mv{v}\in\mathbb{R}^2$,
\begin{equation*}
\mv{v}^{\top}\mv{f}(\boldsymbol\omega)\mv{v}
=
\frac{1}{(2\pi)^2}
\left\|\mv{R}^{\top}\mv{H}(\boldsymbol\omega)^{-1}\mv{v}\right\|^2
>0,
\end{equation*}
so $\mv{f}(\boldsymbol\omega)$ is positive definite at every
frequency. Since $(\mv{R}\mv{R}^{\top})_{kk}=1$, its diagonal
entries are
$f_{kk}(\boldsymbol\omega)=
c_k(\kappa_k^2+\|\boldsymbol\omega\|^2)^{-2q_k}$, where
$c_k := \bigl((2\pi)^2\tau_k^2\bigr)^{-1}$. Polar coordinates give
\begin{align*}
\int_{\mathbb{R}^2} f_{kk}(\boldsymbol\omega)\,d\boldsymbol\omega
&=
2\pi c_k\int_0^\infty
s(\kappa_k^2+s^2)^{-2q_k}\,ds =
\frac{\pi c_k\kappa_k^{2-4q_k}}{2q_k-1}
<\infty,
\end{align*}
because $2q_k=\nu_{k}+1>1$. Positive definiteness also gives
$|f_{12}|^2<f_{11}f_{22}$, and hence
 $\mv{f}$ is entrywise integrable. By Cram\'{e}r's theorem,
its inverse Fourier transform is a valid stationary covariance function
\citep[Theorem~1]{kleiber2017coherence}.

It remains to invert the
cross-spectrum
\begin{equation*}
f_{12}(\boldsymbol\omega)
=
\frac{\rho_S}{(\tau_1\tau_2)(2\pi)^2}
\frac{1}
{(\kappa_1^2+\|\boldsymbol\omega\|^2)^{q_1}
(\kappa_2^2+\|\boldsymbol\omega\|^2)^{q_2}}.
\end{equation*}
Set $A := \kappa_1^2+\|\boldsymbol\omega\|^2$ and
$B := \kappa_2^2+\|\boldsymbol\omega\|^2$. The calculation proceeds
in two steps. First, combine the two denominator factors. For
$q_1,q_2>0$, Euler's beta integral
\citep[equation~(5.12.1)]{olver2010nist} is
\begin{equation*}
\int_0^1 t^{q_1-1}(1-t)^{q_2-1}\,dt
=
\frac{\Gamma(q_1)\Gamma(q_2)}{\Gamma(q_1+q_2)}.
\end{equation*}
Since $A,B>0$, the map
$t := uA/\{uA+(1-u)B\}$ is a bijection of $(0,1)$. Using this
change of variable gives
\begin{equation*}
\frac{1}{A^{q_1}B^{q_2}}
=
\frac{\Gamma(q_1+q_2)}{\Gamma(q_1)\Gamma(q_2)}
\int_0^1
\frac{u^{q_1-1}(1-u)^{q_2-1}}
{\{uA+(1-u)B\}^{q_1+q_2}}\,du.
\end{equation*}
Since $uA+(1-u)B=\kappa_u^2+\|\boldsymbol\omega\|^2$, where
$\kappa_u^2 := u\kappa_1^2+(1-u)\kappa_2^2$, and
$q_1+q_2=\bar\nu+1$, substituting this identity into the
cross-spectrum gives
\begin{equation*}
f_{12}(\boldsymbol\omega)
=
\frac{\rho_S}{\tau_1\tau_2}
\frac{\Gamma(\bar\nu+1)}{\Gamma(q_1)\Gamma(q_2)}
\int_0^1
\frac{u^{q_1-1}(1-u)^{q_2-1}}
{(2\pi)^2
 (\kappa_u^2+\|\boldsymbol\omega\|^2)^{\bar\nu+1}}\,du.
\end{equation*}

Second, the inverse Fourier transform of the 
Mat\'{e}rn spectral density is
\begin{equation*}
\frac{1}{(2\pi)^2}
\int_{\mathbb{R}^2}
\frac{\exp(i\boldsymbol\omega^\top\mathbf{h})}
{(\kappa^2+\|\boldsymbol\omega\|^2)^{\bar\nu+1}}\,
d\boldsymbol\omega
=
\frac{1}{4\pi\bar\nu\,\kappa^{2\bar\nu}}
\mathcal{M}(\kappa r;\bar\nu),
\qquad r := \|\mathbf{h}\|.
\end{equation*}
Applying it
with $\kappa=\kappa_u$, and using
$\Gamma(\bar\nu+1)=\bar\nu\Gamma(\bar\nu)$ gives
\begin{equation*}
C_{12}(r)
=
\frac{\rho_S}{\tau_1\tau_2}
\frac{\Gamma(\bar\nu)}
{4\pi\Gamma(q_1)\Gamma(q_2)}
\int_0^1
u^{q_1-1}(1-u)^{q_2-1}
\kappa_u^{-2\bar\nu}
\mathcal{M}(\kappa_u r;\bar\nu)\,du .
\end{equation*}
Substituting the variance normalisation for $\tau_1$ and $\tau_2$
gives \eqref{eq:biv-matern-crosscov}.
\end{proof}

\begin{proof}[proof of corollary~\ref{cor:biv-matern-special-cases}]
If $\kappa_1=\kappa_2=\kappa$, then $\kappa_u=\kappa$ for every
$u\in(0,1)$. The integral in \eqref{eq:biv-matern-crosscov} is
therefore
\begin{equation*}
\kappa^{-2\bar\nu}\mathcal{M}(\kappa r;\bar\nu)
\int_0^1 u^{q_1-1}(1-u)^{q_2-1}\,du
=
\kappa^{-2\bar\nu}\mathcal{M}(\kappa r;\bar\nu)
\frac{\Gamma(q_1)\Gamma(q_2)}
{\Gamma(\bar\nu+1)}.
\end{equation*}
This gives the closed form \eqref{eq:biv-matern-common-range} for the cross-covariance.

Now let $\nu_1=\nu_2=1$, so that $q_1=q_2=1$. For
$\kappa_1\neq\kappa_2$, the denominator of the cross-spectrum in
the proof of Proposition~\ref{prop:biv-matern-validity} satisfies
\begin{equation*}
\frac{1}{(\kappa_1^2+\|\boldsymbol\omega\|^2)
(\kappa_2^2+\|\boldsymbol\omega\|^2)}
=
\frac{1}{\kappa_2^2-\kappa_1^2}
\left\{
\frac{1}{\kappa_1^2+\|\boldsymbol\omega\|^2}
-
\frac{1}{\kappa_2^2+\|\boldsymbol\omega\|^2}
\right\}.
\end{equation*}
The two-dimensional inverse Fourier transform of
$(\kappa^2+\|\boldsymbol\omega\|^2)^{-1}$ is
$(2\pi)^{-1}K_0(\kappa r)$.
Therefore, for $r>0$,
\begin{equation*}
C_{12}(r)
=
\frac{\rho_S}{2\pi\tau_1\tau_2}
\frac{K_0(\kappa_1r)-K_0(\kappa_2r)}
{\kappa_2^2-\kappa_1^2}.
\end{equation*}
Since $\nu_1=\nu_2=1$, the variance normalisations give
$(\tau_1\tau_2)^{-1}=4\pi\sigma_1\sigma_2\kappa_1\kappa_2$,
which proves \eqref{eq:biv-matern-unit-smoothness}.
\end{proof}

\begin{corollary}[Colocated correlation]
\label{cor:biv-matern-correlation}
Under the conditions of
Proposition~\ref{prop:biv-matern-validity}, define
$\rho:=\operatorname{Corr}\{S_1(\mv{x}),S_2(\mv{x})\}$. Then
\begin{equation*}
\rho=\rho_S\mathcal{A},
\qquad
\mathcal{A}
:=
\frac{\displaystyle
\int_{\mathbb{R}^2}
\{f_{11}(\boldsymbol\omega)f_{22}(\boldsymbol\omega)\}^{1/2}\,
d\boldsymbol\omega}{\sigma_1\sigma_2}
\in(0,1].
\end{equation*}
Thus, for fixed marginals, $\rho\in(-\mathcal{A},\mathcal{A})$, with
$\mathcal{A}=1$ if and only if $\kappa_1=\kappa_2$ and $\nu_1=\nu_2$.
\end{corollary}

\begin{proof}
The spectral-density representation obtained in the proof of
Proposition~\ref{prop:biv-matern-validity} is
\begin{equation*}
\mv{f}(\boldsymbol\omega)
=
\frac{1}{(2\pi)^2}
\mv{H}(\boldsymbol\omega)^{-1}
\mv{R}\mv{R}^{\top}
\mv{H}(\boldsymbol\omega)^{-1},
\end{equation*}
where $\mv{H}(\boldsymbol\omega)$ is diagonal and positive and
$(\mv{R}\mv{R}^{\top})_{12}=\rho_S$. Hence
$f_{12}=\rho_S(f_{11}f_{22})^{1/2}$, so $\rho_S$ is the constant signed
spectral correlation. Integrating this identity and using
$\int f_{kk}(\boldsymbol\omega)\,d\boldsymbol\omega=\sigma_k^2$ gives
$\rho=\rho_S\mathcal{A}$. The Cauchy--Schwarz inequality gives
$\mathcal{A}\leq1$, with equality if and only if the marginal spectra are
proportional. After variance normalisation this is equivalent to identical
spectral shapes, completing the proof.
\end{proof}

\section{Additional simulation results}
\label{app:sim_results}

\begin{table}[H]
\centering
\footnotesize
\caption{Median bias (MB) and median absolute error (MAE) of $S$, $T$ and $Y$ (pooled across both antigens), comparing the separate univariate models (Univ.) and the joint bivariate model (Joint), across the 9 $(\rho_S,\rho_T)$ scenarios and the 4 variance-ratio values $c$, sample size $n_{L} = 600$.}
\label{tab:bias_mae_redhot_n600_wide}
\resizebox{\textwidth}{!}{%
\begin{tabular}{lllrrrrrrrrrrrr}
\toprule
$c$ & $\rho_S$ & $\rho_T$ & \multicolumn{4}{c}{$S$} & \multicolumn{4}{c}{$T$} & \multicolumn{4}{c}{$Y$} \\
\cmidrule(lr){4-7} \cmidrule(lr){8-11} \cmidrule(lr){12-15}
 & & & \multicolumn{2}{c}{MB} & \multicolumn{2}{c}{MAE} & \multicolumn{2}{c}{MB} & \multicolumn{2}{c}{MAE} & \multicolumn{2}{c}{MB} & \multicolumn{2}{c}{MAE} \\
\cmidrule(lr){4-5} \cmidrule(lr){6-7} \cmidrule(lr){8-9} \cmidrule(lr){10-11} \cmidrule(lr){12-13} \cmidrule(lr){14-15}
 & & & Univ. & Joint & Univ. & Joint & Univ. & Joint & Univ. & Joint & Univ. & Joint & Univ. & Joint \\
\midrule
\multirow{9}{*}{1} & 0.1 & 0.1 & -0.006 & -0.015 & 0.718 & 0.704 & -0.085 & 0.004 & 0.206 & 0.161 & 0.016 & 0.006 & 1.111 & 1.104 \\
 & 0.1 & 0.5 & -0.004 & -0.008 & 0.722 & 0.702 & -0.085 & 0.004 & 0.206 & 0.160 & 0.013 & 0.005 & 1.103 & 1.099 \\
 & 0.1 & 0.9 & -0.003 & -0.009 & 0.713 & 0.678 & -0.088 & -0.005 & 0.207 & 0.157 & 0.013 & -0.002 & 1.096 & 1.091 \\
 & 0.5 & 0.1 & -0.016 & -0.031 & 0.709 & 0.677 & -0.085 & 0.004 & 0.205 & 0.160 & 0.018 & 0.004 & 1.107 & 1.099 \\
 & 0.5 & 0.5 & -0.025 & -0.028 & 0.706 & 0.682 & -0.084 & 0.003 & 0.206 & 0.159 & 0.011 & 0.005 & 1.100 & 1.094 \\
 & 0.5 & 0.9 & -0.010 & -0.025 & 0.704 & 0.669 & -0.086 & -0.005 & 0.205 & 0.157 & 0.013 & -0.002 & 1.092 & 1.088 \\
 & 0.9 & 0.1 & -0.016 & -0.040 & 0.692 & 0.627 & -0.085 & 0.004 & 0.205 & 0.159 & 0.017 & 0.006 & 1.104 & 1.093 \\
 & 0.9 & 0.5 & 0.000 & -0.013 & 0.690 & 0.639 & -0.084 & 0.005 & 0.204 & 0.158 & 0.016 & 0.005 & 1.098 & 1.088 \\
 & 0.9 & 0.9 & -0.020 & -0.036 & 0.690 & 0.633 & -0.088 & -0.005 & 0.205 & 0.156 & 0.015 & 0.001 & 1.089 & 1.082 \\
\midrule
\multirow{9}{*}{0.5} & 0.1 & 0.1 & -0.010 & -0.009 & 0.716 & 0.700 & -0.086 & 0.004 & 0.203 & 0.161 & 0.010 & 0.008 & 1.090 & 1.085 \\
 & 0.1 & 0.5 & -0.012 & -0.017 & 0.713 & 0.698 & -0.085 & 0.007 & 0.203 & 0.159 & 0.011 & 0.005 & 1.084 & 1.078 \\
 & 0.1 & 0.9 & -0.010 & -0.009 & 0.712 & 0.670 & -0.090 & -0.005 & 0.205 & 0.157 & 0.011 & -0.001 & 1.075 & 1.070 \\
 & 0.5 & 0.1 & -0.000 & -0.022 & 0.702 & 0.668 & -0.084 & 0.004 & 0.202 & 0.160 & 0.014 & 0.004 & 1.085 & 1.079 \\
 & 0.5 & 0.5 & -0.009 & -0.016 & 0.701 & 0.675 & -0.087 & 0.004 & 0.203 & 0.158 & 0.014 & 0.003 & 1.080 & 1.073 \\
 & 0.5 & 0.9 & -0.014 & -0.019 & 0.698 & 0.664 & -0.091 & -0.004 & 0.204 & 0.156 & 0.012 & -0.000 & 1.073 & 1.062 \\
 & 0.9 & 0.1 & -0.020 & -0.040 & 0.686 & 0.624 & -0.087 & 0.002 & 0.203 & 0.158 & 0.014 & 0.005 & 1.081 & 1.071 \\
 & 0.9 & 0.5 & -0.012 & -0.031 & 0.688 & 0.634 & -0.087 & 0.005 & 0.202 & 0.158 & 0.011 & 0.006 & 1.072 & 1.067 \\
 & 0.9 & 0.9 & -0.024 & -0.040 & 0.684 & 0.625 & -0.085 & -0.005 & 0.201 & 0.155 & 0.012 & -0.001 & 1.068 & 1.060 \\
\midrule
\multirow{9}{*}{0.25} & 0.1 & 0.1 & -0.017 & -0.017 & 0.708 & 0.697 & -0.088 & 0.001 & 0.202 & 0.161 & 0.010 & 0.005 & 1.081 & 1.074 \\
 & 0.1 & 0.5 & -0.006 & -0.003 & 0.715 & 0.696 & -0.086 & 0.002 & 0.201 & 0.160 & 0.011 & 0.006 & 1.073 & 1.069 \\
 & 0.1 & 0.9 & -0.003 & -0.007 & 0.705 & 0.668 & -0.089 & -0.004 & 0.201 & 0.157 & 0.008 & 0.000 & 1.063 & 1.055 \\
 & 0.5 & 0.1 & -0.019 & -0.019 & 0.697 & 0.670 & -0.087 & 0.002 & 0.201 & 0.160 & 0.010 & 0.006 & 1.076 & 1.069 \\
 & 0.5 & 0.5 & -0.011 & -0.018 & 0.697 & 0.673 & -0.089 & 0.001 & 0.201 & 0.159 & 0.012 & 0.007 & 1.069 & 1.063 \\
 & 0.5 & 0.9 & -0.017 & -0.021 & 0.691 & 0.652 & -0.089 & -0.004 & 0.200 & 0.156 & 0.013 & 0.004 & 1.058 & 1.053 \\
 & 0.9 & 0.1 & -0.003 & -0.023 & 0.680 & 0.622 & -0.090 & 0.001 & 0.201 & 0.159 & 0.015 & 0.006 & 1.073 & 1.061 \\
 & 0.9 & 0.5 & -0.027 & -0.039 & 0.685 & 0.632 & -0.086 & 0.001 & 0.200 & 0.158 & 0.011 & 0.006 & 1.067 & 1.059 \\
 & 0.9 & 0.9 & -0.019 & -0.024 & 0.678 & 0.619 & -0.089 & -0.004 & 0.201 & 0.156 & 0.010 & 0.001 & 1.055 & 1.050 \\
\midrule
\multirow{9}{*}{0.1} & 0.1 & 0.1 & -0.008 & -0.010 & 0.710 & 0.694 & -0.087 & -0.001 & 0.199 & 0.161 & 0.009 & 0.003 & 1.072 & 1.065 \\
 & 0.1 & 0.5 & -0.014 & -0.014 & 0.710 & 0.692 & -0.091 & -0.001 & 0.201 & 0.160 & 0.010 & 0.003 & 1.064 & 1.062 \\
 & 0.1 & 0.9 & -0.006 & -0.003 & 0.701 & 0.660 & -0.092 & -0.005 & 0.201 & 0.157 & 0.009 & 0.000 & 1.056 & 1.049 \\
 & 0.5 & 0.1 & -0.011 & -0.021 & 0.695 & 0.669 & -0.089 & -0.003 & 0.200 & 0.161 & 0.011 & 0.005 & 1.066 & 1.060 \\
 & 0.5 & 0.5 & -0.010 & -0.014 & 0.695 & 0.673 & -0.086 & -0.001 & 0.198 & 0.159 & 0.013 & 0.005 & 1.062 & 1.057 \\
 & 0.5 & 0.9 & -0.009 & -0.016 & 0.691 & 0.654 & -0.090 & -0.004 & 0.200 & 0.156 & 0.010 & 0.002 & 1.052 & 1.046 \\
 & 0.9 & 0.1 & -0.006 & -0.026 & 0.680 & 0.624 & -0.089 & 0.001 & 0.199 & 0.159 & 0.014 & 0.005 & 1.064 & 1.055 \\
 & 0.9 & 0.5 & -0.008 & -0.018 & 0.679 & 0.630 & -0.088 & 0.002 & 0.198 & 0.158 & 0.012 & 0.007 & 1.058 & 1.052 \\
 & 0.9 & 0.9 & -0.008 & -0.018 & 0.677 & 0.614 & -0.088 & -0.003 & 0.198 & 0.155 & 0.008 & 0.002 & 1.050 & 1.040 \\
\bottomrule
\end{tabular}
}
\end{table}

\begin{table}[H]
\centering
\footnotesize
\caption{Median bias (MB) and median absolute error (MAE) of $S$, $T$ and $Y$ (pooled across both antigens), comparing the separate univariate models (Univ.) and the joint bivariate model (Joint), across the 9 $(\rho_S,\rho_T)$ scenarios and the 4 variance-ratio values $c$, sample size $n_{L} = 1200$.}
\label{tab:bias_mae_redhot_n1200_wide}
\resizebox{\textwidth}{!}{%
\begin{tabular}{lllrrrrrrrrrrrr}
\toprule
$c$ & $\rho_S$ & $\rho_T$ & \multicolumn{4}{c}{$S$} & \multicolumn{4}{c}{$T$} & \multicolumn{4}{c}{$Y$} \\
\cmidrule(lr){4-7} \cmidrule(lr){8-11} \cmidrule(lr){12-15}
 & & & \multicolumn{2}{c}{MB} & \multicolumn{2}{c}{MAE} & \multicolumn{2}{c}{MB} & \multicolumn{2}{c}{MAE} & \multicolumn{2}{c}{MB} & \multicolumn{2}{c}{MAE} \\
\cmidrule(lr){4-5} \cmidrule(lr){6-7} \cmidrule(lr){8-9} \cmidrule(lr){10-11} \cmidrule(lr){12-13} \cmidrule(lr){14-15}
 & & & Univ. & Joint & Univ. & Joint & Univ. & Joint & Univ. & Joint & Univ. & Joint & Univ. & Joint \\
\midrule
\multirow{9}{*}{1} & 0.1 & 0.1 & -0.009 & -0.007 & 0.666 & 0.655 & -0.084 & 0.004 & 0.201 & 0.161 & 0.013 & 0.004 & 1.110 & 1.106 \\
 & 0.1 & 0.5 & -0.020 & -0.017 & 0.665 & 0.650 & -0.084 & 0.002 & 0.200 & 0.159 & 0.010 & 0.004 & 1.103 & 1.100 \\
 & 0.1 & 0.9 & -0.014 & -0.013 & 0.660 & 0.623 & -0.090 & -0.006 & 0.202 & 0.157 & 0.010 & -0.000 & 1.096 & 1.089 \\
 & 0.5 & 0.1 & -0.014 & -0.022 & 0.649 & 0.630 & -0.084 & 0.004 & 0.201 & 0.159 & 0.014 & 0.004 & 1.106 & 1.100 \\
 & 0.5 & 0.5 & -0.025 & -0.030 & 0.650 & 0.634 & -0.085 & 0.002 & 0.201 & 0.158 & 0.012 & 0.003 & 1.101 & 1.095 \\
 & 0.5 & 0.9 & -0.024 & -0.024 & 0.645 & 0.614 & -0.094 & -0.006 & 0.203 & 0.156 & 0.012 & -0.002 & 1.092 & 1.085 \\
 & 0.9 & 0.1 & -0.015 & -0.030 & 0.630 & 0.576 & -0.087 & 0.003 & 0.201 & 0.159 & 0.015 & 0.001 & 1.101 & 1.093 \\
 & 0.9 & 0.5 & -0.020 & -0.030 & 0.630 & 0.582 & -0.085 & 0.002 & 0.200 & 0.157 & 0.012 & 0.001 & 1.096 & 1.090 \\
 & 0.9 & 0.9 & -0.033 & -0.048 & 0.629 & 0.579 & -0.091 & -0.007 & 0.202 & 0.156 & 0.014 & -0.001 & 1.089 & 1.082 \\
\midrule
\multirow{9}{*}{0.5} & 0.1 & 0.1 & -0.021 & -0.016 & 0.661 & 0.650 & -0.085 & -0.001 & 0.198 & 0.160 & 0.009 & 0.005 & 1.088 & 1.085 \\
 & 0.1 & 0.5 & -0.023 & -0.020 & 0.660 & 0.647 & -0.086 & 0.000 & 0.198 & 0.159 & 0.010 & 0.004 & 1.082 & 1.081 \\
 & 0.1 & 0.9 & -0.020 & -0.020 & 0.657 & 0.614 & -0.094 & -0.007 & 0.201 & 0.157 & 0.010 & -0.003 & 1.076 & 1.068 \\
 & 0.5 & 0.1 & -0.016 & -0.020 & 0.642 & 0.626 & -0.083 & 0.000 & 0.197 & 0.159 & 0.012 & 0.004 & 1.086 & 1.079 \\
 & 0.5 & 0.5 & -0.020 & -0.027 & 0.642 & 0.629 & -0.087 & -0.000 & 0.197 & 0.158 & 0.010 & 0.004 & 1.080 & 1.074 \\
 & 0.5 & 0.9 & -0.022 & -0.021 & 0.640 & 0.608 & -0.092 & -0.007 & 0.199 & 0.157 & 0.011 & 0.000 & 1.072 & 1.066 \\
 & 0.9 & 0.1 & -0.017 & -0.032 & 0.621 & 0.569 & -0.084 & -0.000 & 0.196 & 0.158 & 0.013 & 0.005 & 1.079 & 1.074 \\
 & 0.9 & 0.5 & -0.032 & -0.044 & 0.625 & 0.582 & -0.084 & 0.000 & 0.196 & 0.157 & 0.010 & 0.002 & 1.074 & 1.071 \\
 & 0.9 & 0.9 & -0.026 & -0.033 & 0.622 & 0.570 & -0.093 & -0.007 & 0.200 & 0.156 & 0.009 & 0.002 & 1.065 & 1.060 \\
\midrule
\multirow{9}{*}{0.25} & 0.1 & 0.1 & -0.020 & -0.019 & 0.657 & 0.648 & -0.089 & -0.005 & 0.198 & 0.161 & 0.008 & 0.004 & 1.079 & 1.075 \\
 & 0.1 & 0.5 & -0.014 & -0.011 & 0.660 & 0.642 & -0.090 & -0.004 & 0.197 & 0.160 & 0.009 & 0.004 & 1.073 & 1.068 \\
 & 0.1 & 0.9 & -0.018 & -0.014 & 0.654 & 0.610 & -0.097 & -0.008 & 0.200 & 0.158 & 0.010 & -0.000 & 1.064 & 1.058 \\
 & 0.5 & 0.1 & -0.024 & -0.031 & 0.639 & 0.622 & -0.090 & -0.005 & 0.198 & 0.161 & 0.011 & 0.004 & 1.074 & 1.071 \\
 & 0.5 & 0.5 & -0.022 & -0.029 & 0.641 & 0.626 & -0.091 & -0.003 & 0.197 & 0.159 & 0.011 & 0.004 & 1.068 & 1.064 \\
 & 0.5 & 0.9 & -0.025 & -0.024 & 0.638 & 0.604 & -0.096 & -0.007 & 0.199 & 0.157 & 0.009 & 0.000 & 1.060 & 1.054 \\
 & 0.9 & 0.1 & -0.027 & -0.041 & 0.621 & 0.568 & -0.091 & -0.006 & 0.197 & 0.160 & 0.010 & 0.003 & 1.070 & 1.064 \\
 & 0.9 & 0.5 & -0.024 & -0.039 & 0.621 & 0.577 & -0.092 & -0.003 & 0.197 & 0.158 & 0.012 & 0.004 & 1.063 & 1.058 \\
 & 0.9 & 0.9 & -0.034 & -0.038 & 0.620 & 0.566 & -0.095 & -0.007 & 0.199 & 0.156 & 0.011 & -0.001 & 1.055 & 1.048 \\
\midrule
\multirow{9}{*}{0.1} & 0.1 & 0.1 & -0.017 & -0.017 & 0.656 & 0.648 & -0.091 & -0.009 & 0.197 & 0.162 & 0.008 & 0.003 & 1.072 & 1.068 \\
 & 0.1 & 0.5 & -0.015 & -0.012 & 0.657 & 0.639 & -0.094 & -0.007 & 0.198 & 0.161 & 0.008 & 0.004 & 1.066 & 1.065 \\
 & 0.1 & 0.9 & -0.019 & -0.011 & 0.655 & 0.610 & -0.096 & -0.007 & 0.198 & 0.158 & 0.009 & 0.000 & 1.058 & 1.050 \\
 & 0.5 & 0.1 & -0.022 & -0.028 & 0.638 & 0.622 & -0.093 & -0.010 & 0.197 & 0.161 & 0.010 & 0.002 & 1.067 & 1.063 \\
 & 0.5 & 0.5 & -0.023 & -0.030 & 0.637 & 0.625 & -0.095 & -0.007 & 0.197 & 0.159 & 0.010 & 0.003 & 1.061 & 1.058 \\
 & 0.5 & 0.9 & -0.027 & -0.030 & 0.636 & 0.602 & -0.096 & -0.007 & 0.197 & 0.157 & 0.008 & 0.001 & 1.054 & 1.047 \\
 & 0.9 & 0.1 & -0.021 & -0.035 & 0.618 & 0.569 & -0.095 & -0.010 & 0.197 & 0.161 & 0.012 & 0.004 & 1.064 & 1.057 \\
 & 0.9 & 0.5 & -0.031 & -0.043 & 0.619 & 0.577 & -0.092 & -0.008 & 0.195 & 0.159 & 0.011 & 0.002 & 1.056 & 1.052 \\
 & 0.9 & 0.9 & -0.033 & -0.045 & 0.616 & 0.564 & -0.097 & -0.007 & 0.197 & 0.157 & 0.009 & -0.001 & 1.048 & 1.043 \\
\bottomrule
\end{tabular}
}
\end{table}

\end{document}